\documentclass[11pt]{article}
\usepackage{amsmath}
\usepackage{amsthm}
\usepackage{amssymb}
\usepackage{algorithm}
\usepackage{color}
\usepackage[english]{babel}
\usepackage{graphicx}
\usepackage{grffile}
\usepackage{wrapfig,epsfig}
\usepackage{epstopdf}
\usepackage{url}
\usepackage{color}
\usepackage{epstopdf}
\usepackage{algpseudocode}
\usepackage[T1]{fontenc}
\usepackage{bbm}
\usepackage{comment}
\usepackage{dsfont}
\usepackage{thm-restate}
\usepackage{bm}
\usepackage{tcolorbox}
\usepackage{subcaption}

\usepackage{enumitem}

\usepackage{tikz}
\usetikzlibrary{arrows}

\usepackage[margin=1in]{geometry}

\graphicspath{{./figs/}}
\usepackage{mathtools}

\newtheorem{theorem}{Theorem}[section]
 
\newtheorem{lemma}[theorem]{Lemma}
\newtheorem{definition}[theorem]{Definition}

\newcommand{\wh}{\widehat}
\newcommand{\wt}{\widetilde}

\newcommand{\eps}{\epsilon}

\newcommand{\R}{\mathbb{R}}

\renewcommand{\varepsilon}{\epsilon}
\renewcommand{\tilde}{\wt}
\renewcommand{\hat}{\wh}

\renewcommand{\eps}{\epsilon}

\newcommand{\bs}{\mathsf{bs}}

\newcommand{\nil}{\mathsf{nil}}
\newcommand{\ms}{s^{\mathsf{m}}}

\newcommand{\D}{\mathcal{D}}

\newcommand{\mB}{\mathcal{B}}

\newcommand{\mN}{\mathcal{N}}
\newcommand{\sE}{\mathsf{E}}

\DeclareMathOperator*{\E}{{\mathbb{E}}}

\DeclareMathOperator{\poly}{poly}

\algnewcommand\algorithmicforeach{\textbf{for each}}
\algdef{S}[FOR]{ForEach}[1]{\algorithmicforeach\ #1\ \algorithmicdo}

\makeatletter
\newcommand*{\RN}[1]{\expandafter\@slowromancap\romannumeral #1@}
\makeatother

\title{Complexity of Equilibria in First-Price Auctions\\ under General Tie-Breaking Rules }
\author{}
\author{  
Xi Chen \\ Columbia University \\ \texttt{xichen@cs.columbia.edu}
\and Binghui Peng \\ Columbia University \\ \texttt{bp2601@columbia.edu}
}
\date{}

\begin{document}
\maketitle

\begin{abstract}
We study the complexity of finding an approximate (pure) Bayesian Nash equilibrium in a first-price auction with common priors when the tie-breaking rule is part of the input.
We show that the problem is PPAD-complete even when the tie-breaking rule is trilateral (i.e., it specifies item allocations when no more than three bidders are in tie, and adopts the uniform tie-breaking rule otherwise). This is the first hardness result for equilibrium computation in first-price auctions with common priors.
On the positive side, we give a PTAS for the problem under the uniform tie-breaking rule.
\end{abstract}

\setcounter{page}{0}
\thispagestyle{empty}

\newpage
\section{Introduction}

\def\calD{\mathcal{D}}

First-price auction is arguably the most commonly used auction format in practice \cite{vickrey1961counterspeculation,roughgarden2010algorithmic,conitzer2022pacing}, in which the highest bidder wins the item and pays her bid.
First-price auction and its variants have been widely used in online ad auctions:
when a user visits a platform, an auction is run among interested advertisers to determine the ad
to be displayed to the user.
Despite of its simplicity,
first-price auctions are {\em not} incentive compatible ---
it is the most well-known example in auction theory that does not admit a truthful strategy.
This has led to significant effort in economics \cite{lebrun1996existence,lebrun1999first,maskin2000equilibrium,lizzeri2000uniqueness,athey2001single,maskin2003uniqueness,reny2004existence,lebrun2006uniqueness,bergemann2017first} and more recently, in computer science \cite{chawla2013auctions,wang2020bayesian,filos2021complexity}, to understand equilibria of first-price auctions.

In this paper we study the computational complexity of finding a Bayesian Nash equilibrium in a first-price auction.
We consider the following {\em independent common prior} setting.
There is one single item to sell, and $n$ bidders are interested in it. Each bidder has a continuous value distribution $\calD_i$ supported over $[0,1]$. 
The joint value distribution $\calD$ is the product   of $\calD_i$'s. 
While $\calD_i$'s are public, 
each bidder $i$ has a private value $v_i$ for the item drawn from $\calD_i$. 
Each bidder chooses a bidding strategy, which maps her private value to a bid from a {\em discrete} bid space $\mB = \{b_0,b_1 \ldots, b_m\}$. A Bayesian Nash equilibrium is a tuple of bidding strategies, one for each bidder, such that every bidder gets a best response to other bidders' strategies
(see formal definition in Section \ref{sec:pre}, including the two conditions that bidding strategies need to satisfy: no overbidding and monotonicity).

This game between bidders, however, is not fully specified without a tie-breaking rule: how the item is allocated when more than one bidder have the highest bid. 
A variety of tie-breaking rules have been considered in the literature.
The uniform tie-breaking rule, where the item is allocated to one of the winners uniformly at random, has been the most common offset. The other commonly used
tie-breaking rule is to  perform an additional round of Vickrey auction to ensure the existence of equilibria when the bidding space is continuous \cite{lebrun1996existence, deng2022efficiency}.
A recent line of works \cite{syrgkanis2014efficiency,hartline2014price,jin2022first,jin2023price} used monopoly tie-breaking rules that always give the item to one player  when establishing worst-case price-of-anarchy (POA) bounds.

To accommodate 
tie-breaking rules in the problem, we consider the setting where the auctioneer specifies a tie-breaking rule $\Gamma$  to be used in the auction (as part of the input). $\Gamma$  maps each $W\subseteq [n]$ as the set of winners to a distribution $\Gamma(W)$ over $W$ as the allocation of the item to bidders in $W$. While a general tie-breaking rule takes exponentially many entries to describe, our PPAD-hardness result is built upon the succinct family of so-called \emph{trilateral} tie-breaking rules: such a tie-breaking rule $\Gamma$  specifies item allocations when no more than three bidders are in tie, and follows the uniform tie-breaking rule otherwise (when more than three bidders are in tie).
The hardness result rules out the  possibility of an efficient algorithm for finding a Bayesian Nash equilibrium  in a first-price auction when the tie-breaking rule is given as part of the input (unless PPAD is in P). 
We compliment our hardness result with a polynomial time approximation scheme (PTAS) for finding a constant-approximate Bayesian Nash equilibrium under the uniform tie-breaking rule.


\subsection{Our results}

Our main hardness result shows that the problem of finding an $\eps$-approximate Bayesian Nash equilibrium in a first-price auction is PPAD-complete with trilateral tie-breaking.

\begin{restatable}[Computational hardness]{theorem}{Hardness}
\label{thm:hardness}
It is PPAD-complete to find an $\eps$-approximate Bayesian Nash equilibrium in a first-price auction under a trilateral tie-breaking rule for $\eps = 1/\poly(n)$.
\end{restatable}

It is worth pointing out that the hardness above holds even when (1) the bid space $\mathcal{B}$ has size $3$; and (2) the density function of each $\calD_i$ is a piecewise-constant function with no more than four nonzero pieces.

On the positive side, we obtain a PTAS for finding a Bayesian Nash equilibrium in a first-price auction under the uniform tie-breaking rule:

\begin{restatable}[PTAS under uniform tie-breaking]{theorem}{PTAS}
\label{thm:ptas}
For any $\eps > 0$, $n, m \geq 2$, there is an algorithm that finds an $\eps$-approximate Bayesian Nash equilibrium using $O(n^4 \cdot g(1/\eps))$ time under the uniform tie-breaking rule. 
\end{restatable}

Our algorithm works as long as it has oracle access to the CDF of each value distribution $\calD_i$.

\subsection{Related work}

\paragraph{First-price auction}
The study of first-price auction dates back to the seminal work of Vickrey \cite{vickrey1961counterspeculation} in 1960s.
Despite its extremely simple form and a wide range of applications, the incentive has been a central issue and it is perhaps the most well known mechanism that does not admit a truthful strategy. 
A long line of works in the economic literature \cite{riley1981optimal,plum1992characterization,marshall1994numerical,lebrun1996existence,lebrun1999first,maskin2000equilibrium,lizzeri2000uniqueness,athey2001single,maskin2003uniqueness,reny2004existence,lebrun2006uniqueness,chawla2013auctions,bergemann2017first} devote to characterizing the existence, uniqueness and closed-form expression of a pure Bayesian Nash equilibrium (or BNE).
However, the BNE of first-price auction is only well-understood in a few special cases, including when the players have symmetric valuation distributions \cite{chawla2013auctions}, when all players have probability density function bounded above $0$ and atomic probability mass at the lowest points \cite{lebrun2006uniqueness},  when there are only two bidders with uniform valuation distributions \cite{kaplan2012asymmetric} or when the players have discrete value and continuous bidding space and the tie-breaking is performed with an extra round of Vickrey (second-price) auction \cite{wang2020bayesian}. 

A formal study on the computational complexity of equilibria in a first-price auction has been raised by the recent work of \cite{filos2021complexity}, which is most closest to us. \cite{filos2021complexity} examines the computation complexity under a {\em subjective prior}, that is, each bidder has a different belief of other's valuation distribution. They prove the PPAD-completeness and the FIXP-completeness of finding an $\eps$-BNE (for some constant $\eps > 0$) and an exact BNE, under the uniform tie-breaking rule. As we shall explain soon, the techniques to obtain their results are quite different from us. 
It is worth noting that most aforementioned literature are on the common prior setup, and \cite{filos2021complexity} also leaves an open question of characterizing the computational complexity of $\eps$-BNE under the standard setting of independent common prior.
\cite{filos2021complexity} also provides a polynomial time algorithm for finding a high precision BNE for {\em constant} number of players and bids, when the input distribution are piecewise polynomial. Their approach is based on polynomial system solvers and thus different from us. 
The work of \cite{cai2014simultaneous} studies the Bayesian combinatorial auctions, where there are multi-items to sell for multiple bidders. They prove the complexity of Bayesian Nash equilibrium is at least PP-hard (a complexity class between the polynomial hierarchy and PSPACE), the model is quite different, because the agents' valuation could be much more complex, defining over subsets of items.

Other aspects of first-price auction have also been studied in the literature, including the price of anarchy/stability \cite{syrgkanis2013composable,syrgkanis2014efficiency,feldman2013simultaneous,hoy2018tighter,jin2022first,jin2023price} and parameter estimation \cite{guerre2000optimal,cherapanamjeri2022estimation}.


\paragraph{Equilibrium computation}

The complexity class of PPAD (Polynomial Parity Arguments on Directed graphs) was first introduced by Papadimitriou~\cite{papadimitriou1994complexity} to capture one particular genre of total search functions. 
The seminal work \cite{daskalakis2009complexity,chen2009settling} established the PPAD-hardness of normal-form games. The hardness of approximation was settled by subsequent work~\cite{rubinstein2018inapproximability,rubinstein2016settling, deligaks2022pure} in the past few years. 
A broad range of problems have been proved to be PPAD-hard, and notable examples including  equilibrium computation in special but important class of games (win-or-lose game \cite{abbott2005complexity, chen2007approximation}, anonymous game \cite{chen2015complexity}, constant rank game \cite {mehta2014constant}, graphical game \cite{papadimitriou2021public}),
market equilibrium (Arrow-Debreu market~\cite{chen2009spending, chen2009settling-market, vazirani2011market}, non-monotone market \cite{chen2013complexity,rubinstein2019hardness}, Hylland-Zeckhauser scheme \cite{chen2022computational}),
fair division \cite{othman2016complexity, chaudhury2021competitive}, min-max optimization~\cite{daskalakis2021complexity} 
and reinforcement learning \cite{daskalakis2022complexity,jin2022complexity}.

The PTAS is known for anonymous game \cite{daskalakis2015approximate}, which is closely related to our work. 
The \cite{daskalakis2015approximate} presented a $n^{g(m, 1/\eps)} \cdot U$ algorithm for $m$-action $n$-player anonymous games for some exponential function $g$. Here $U$ denotes the number of bits to represent a payoff value in the game. 
Instead, our algorithm finds an $\eps$-BNE of first-price auction with running time $n^4 \cdot g(1/\eps)$, which does not depend on the size of bidding space and the bit-size of the representation of the distributions.
It crucially utilizes the structure of first-price auction in the rounding and searching step, and could have a broader application in auction theory.

\subsection{Technical overview}
The challenge of obtaining the PPAD-hardness arises from two folds.
First, the utility function does not admit a closed-form expression, in terms of other player's strategy. It depends on an exponential number of possible bidding profiles and is computed only via a dynamic programming approach.
Second, the game structure is highly symmetric under the (independent) common prior.
In a first-price auction, the allocation is determined by the entire bidding profile, and each player faces ``almost'' the same set of profile. 
From this perspective, it is more like an anonymous game. 
Perhaps even worse, in an anonymous game, the utility function of each player is different and could be designed for the sake of reduction.
While in a first-price auction, the utility function of each player is the same, and depends only on the allocation probability.
Of course, the general (non-uniform) tie-breaking rule as well as the different valuation distributions could be used for breaking the symmetry.
We note the above challenges are unique to the common prior setting.
In a sharp contrast, in the subjective prior setting \cite{filos2021complexity}, the players' subjective belief could be different. A player could presume most other players
have zero value and submit zero bid, hence, the game is {\em local} and non-symmetric.

To resolve the above challenges, our key ideas are (1) linearizing the allocation probability and expanding a first order approximation of the utility function; and (2) carefully incorporating a (simple) general tie-breaking rule to break the symmetry.

\paragraph{Technical highlight: Linearizing ``everything''}
Given a strategy profile $s$, the distribution over the entire bidding profile (and therefore the allocation probability, the utility, the best response) could be complicated to compute, especially when multiple players submit the highest bid.
To circumvent this issue, we assign a large probability $(1-\delta)$ around value $0$ for all players, for some polynomially small $\delta > 0$.\footnote{This is the reason that our hardness result only applies for (inverse) polynomially small $\eps$.}
By doing this, the probability that a player bids nonzero is small, so one can ignore higher order term.
Concretely, let $p_{i, j}$ be the probability that player $i$ gets the item when bidding $b_j$ given that the other players have strategy $s_{-i}$, and let $\Gamma_i(b_j, s_{-i})$ be the allocation for player $i$ of bidding $b_j$ given other player's strategy $s_{-i}$. 
The immediate advantage is that the allocation probability can be approximated as 
\begin{align}
\Gamma_i(b_j,s_{-i}) \approx (1 - \sum_{i'\in [n]\setminus \{i\}}\sum_{j' > j}p_{i', j'}) + \sum_{i'\in [n]\setminus \{i\}} \Sigma_{i, i'} \cdot  p_{i', j}\label{eq:tie2}
\end{align}
under a {\em bilateral} tie-breaking rule. 
Here $\Sigma \in [0,1]^{n\times n}$ specifies the allocation when there is a tie between a pair of players $(i_1, i_2)$ and satisfies $\Sigma + \Sigma^{\top} = (J - I)$.
At this stage,
it is tempting to use $p_{i, j}$ to encode variables of a generalized circuit problem and the choice of best response to encode constraints. 
In our final construction, we only need three bids $0 = b_{0} < b_1 < b_2$ and the variables are encoded by the jump point $\tau_{i}$ between $b_{1}, b_{2}$ (i.e., when player $i$ bids $b_2$ instead of $b_1$), which has the closed-form expression of
\begin{align}
    \tau_{i} = b_{2} + \frac{\Gamma_i(b_{1}, s_{-i})\cdot (b_{2} - b_1)}{\Gamma_i(b_{2}, s_{-i}) - \Gamma_i(b_{1}, s_{-i})}.\label{eq:tie-4}
\end{align}
Even after the linearization step of Eq.~\eqref{eq:tie2}, the above expression is still quite formidable to handle. 
Our next idea is to restrict the jumping point in a small interval between $(b_2,1)$, and assign only a small total probability mass of $\beta\delta$ over the interval, here $\beta$ is another polynomially small value. There is a (fixed) probability mass of $\delta$ around $b_2$ and $1$. 
One can further perform a first order approximation to Eq.~\eqref{eq:tie-4}, and again linearize the jumping point expression.

\paragraph{Incorporating tie-breaking rule}
Abstracting away some construction details, the above construction reduces the first-price auction from a fix point problem, obeys the following form
\begin{align}
\vec{p} = f(G \vec{p}) \quad \text{where} \quad G = 2\Sigma^{A} - J + I,\label{eq:fix8}
\end{align}
where $\vec{p}$ is the probability of bidding $b_1$ (inside the small interval), $f = (f_1, \ldots, f_n)$ is operated coordinate-wise over $G\vec{p}$, $f_i$ is a monotone function maps from $[a_i,b_i]$ (some fixed interval) to $[0, 1]$, $J$ is the all $1$ matrix and $I$ is the identity matrix.
The fixed point problem is fairly general and subsumes the generalized circuit problem, {\em if} $\Sigma^A$ is an arbitrary matrix in $[0,1]^{n\times n}$. 
Unfortunately, it is not true due to the constraint of $\Sigma^A + (\Sigma^{A})^{\top} = (J- I)$. 
We resolve the issue by adding an extra pivot player. 
The pivot player is guaranteed to bid $b_0 = 0$ and $b_2$ with equal probability of $1/2$. 
From a high level, the pivot player splits the equilibrium computation into two cases,  the case when it bids $b_0$ is similar,  while the case of  bidding $b_2$ introduces another tie-breaking matrix $\Sigma^{B} \in [0,1]^{n \times n}$ among the original players in $[n]$ (hence it becomes a trilateral rule). 
It transforms the fix point problem (i.e., Eq.~\eqref{eq:fix8}) to a more convenient form 
\begin{align}
\vec{p} = f(G' \vec{p}) \quad \text{where} \quad G' = 2\Sigma^{A} + \Sigma^{B} - J + I,\label{eq:fix9}
\end{align}
and one can construct gadgets to reduce from the generalized circuit problem. 
The last step is fairly common and details can be found in Section \ref{sec:ppad}.


\section{Preliminary}
\label{sec:pre}

\paragraph{Notation.} We write $[n]$ to denote $\{1,2,  \ldots, n\}$ and $[n_1:n_2]$ to denote $ \{n_1, n_1+1, \ldots, n_2\}$. Let $\mathsf{1}_{i}$ be an indicator vector -- it equals the all $0$ vector, except the $i$-th coordinate which equals $1$.
Let $\Delta_n$ contains all probability distribution over $[n]$. Given a vector $v \in \R^n$, and an index $i \in [n]$, $v_i$ denotes the $i$-th entry of $v$ while $v_{-i}$ denotes $(v_1, v_2, \ldots, v_{i-1}, v_{i+1}, \ldots, v_{n})$, i.e., all entries except the $i$-th coordinate. We write $x = y \pm \eps$ if $x \in [y-\eps, y+\eps]$. 
Let $J_n \in \R^{n\times n}$ be the $n\times n$ all-$1$ matrix and $I_n$ be the $n\times n$ identity matrix.

\subsection{Model}
In a Bayesian first-price auction (FPA), there is one single item to sell and it is specified by a tuple $(\mN, \mB, \D, \Gamma)$, where $\mN = [n]$ is the set of players, $\mB$ is the bid space, $\D$ is the value distribution and $\Gamma$ is the tie-breaking rule.
For each play $i \in \mN$, it has a private value $v_i$ of the item that is drawn from a 
(continuous) distribution $\D_{i}$ supported over $[0, 1]$ (written as $v_i \sim \D_{i}$).
We consider the standard {\em independent common prior} setting --- the joint value distribution $\D = \D_1 \times \cdots \times \D_{n}$ is the product distribution of $\{\D_i\}_{i \in [n]}$ and we assume the value profile $v = (v_1, \ldots, v_{n}) \in [0,1]^{n}$ is drawn from $\D$.
Let $\mB = \{b_0, b_1, \ldots, b_m\} \subset  [0, 1]$ be the bid space, where $0 = b_0 < b_1 < \cdots < b_m \leq 1$. 

In a first-price (sealed-bid) auction, each bidder $i$ submits a bid $\beta_i \in \mB$ simultaneously to the seller. 
The seller assigns the item to the winning player $i^{*}$ which submits the highest bid, and charges $i^{*}$ a payment equals to its bid $\beta_{i^{*}}$.

\paragraph{Allocation and tie-breaking rule.} When there are multiple players submitting the same highest bid, the seller assigns and charges the item to one of those winning players, following a pre-described {\em tie-breaking} rule $\Gamma$. 
A tie breaking rule $\Gamma: \{0, 1\}^n \rightarrow \Delta_n$ maps a set of winning players $W \subseteq [n]$ to an allocation profile $\Gamma(W) \in \Delta_n$ supported on $W$ that specifies the winning probability of each player $i \in W$ as $\Gamma_i(W)$.
Formally, given a bidding profile $\beta \in \mB^n$, the set of winning players $W(\beta)$ are those who submit the highest bids \[
W(\beta) = \left\{i  \in [n]: \beta_{i } = \max_{j\in [n]}\beta_{j}\right\}.
\]
The tie breaking rule $\Gamma(W(\beta)) \in \Delta_n$ specifies the winning probability of each player in $W(\beta)$ and $\Gamma_i(W(\beta))$ is the probability that the bidder $i$ obtains the item under the bidding profile $\beta$.
The tie-breaking rule needs to satisfy (1) $\Gamma_i(W(b)) > 0$ only if $i \in W(\beta)$, i.e., the item is assigned only to  players with the highest bid; and (2) $\sum_{i\in [n]}\Gamma_i(W(\beta)) = 1$, i.e., the total allocation is $1$.
When there is no confusion, we also abbreviate $\Gamma(\beta) = \Gamma(W(\beta))$.

It is known that the tie-breaking rule plays a subtle yet critical rule on the equilibrium of Bayesian FPA. 
Our hardness result is built upon the {\em trilateral} tie-breaking rule, a simple generalization of the commonly used uniform tie-breaking method.
\begin{definition}[Trilateral tie-breaking]
A \emph{trilateral} tie-breaking rule $\Gamma$ is 
  specified by the following tuples of nonnegative numbers 
$$\big(w_{i,j}:1\le i<j\le n\big)\quad \text{and}\quad
\left(\sigma_{i,j,k}^{(1)},\sigma_{i,j,k}^{(2)}:1\le i<j<k\le n\right)
  $$
such that $w_{i,j}\le 1$ and $\sigma_{i,j,k}^{(1)}+\sigma_{i,j,k}^{(2)}\le 1$.
Given a bidding profile $\beta\in \mB^n$ and the winning set $W(\beta)$, the item is distributed according to $\Gamma$ as follows
\begin{flushleft}\begin{enumerate}
    \item If $W(\beta)=\{i\}$ for some $i\in [n]$, then $\Gamma_i(\beta)=1$;
    \item If $W(\beta)=\{i,j\}$ for some $1\le i<j\le n$, then $\Gamma_i(\beta)=w_{i,j}$ and $\Gamma_j(\beta)=1-w_{i,j}$;
    \item If $W(\beta)=\{i,j,k\}$ for some $1\le i<j<k\le n$, then $\Gamma_i(\beta)=\sigma_{i,j,k}^{(1)}$,
    $\Gamma_j(\beta)=\sigma_{i,j,k}^{(2)}$ and 
    $\Gamma_k(\beta)=1-\sigma_{i,j,k}^{(1)}-\sigma_{i,j,k}^{(2)}$; and 
    \item When $| W(\beta)|\ge 4$, the item is evenly distributed among players in $W(\beta)$. \footnote{We note our hardness result actually holds regardless of the tie-breaking rule among more than $3$ players (i.e., not necessarily uniform).}
\end{enumerate}\end{flushleft}

\end{definition}


\paragraph{Equilibrium and strategy}
Given a tie-breaking rule $\Gamma$ and a bidding profile $\beta = (\beta_1, \ldots, \beta_n)$, the {\em ex-post} utility of a bidder $i$ is given by 
\begin{align*}
    u_i(v_i; \beta_i, \beta_{-i}) = (v_i - \beta_i) \cdot \Gamma_i(\beta).
\end{align*}
A strategy $s_i: [0, 1] \rightarrow \mB$ of player $i$ is a map from her (private) value $v_i$ to a bid $s(v_i) \in \mB$, with the following two properties:
\begin{itemize}
    \item {\bf No overbidding}. A player never submits a bid larger than her private value, i.e., $s_i(v_i) \leq v_i$ for all $v_i \in [0,1]$.
    \item {\bf Monotonicity.}  $s_i$ is a non-decreasing function.
\end{itemize}
These are common assumptions in the literature of first-price auction \cite{maskin2003uniqueness,lebrun2006uniqueness,filos2021complexity} and they rule out spurious equilibria in Bayesian auctions \cite{cai2014simultaneous}. 
Due to the monotonicity assumption, one can write a strategy $s_i$ as $m$ thresholds $0 \leq \tau_{i,1} \leq \cdots \leq \tau_{i, m} \leq 1$, where the player $i$  bids $b_j$ in the interval $(\tau_{i, j}, \tau_{i, j+ 1}]$\footnote{If the valuation distribution contains a point mass, then the strategy might be randomized at the point mass.}. Here we set by default $\tau_{i, 0}=0$ and $\tau_{i, m+1} = 1$.

The $\eps$-approximate Bayesian Nash equilibrium ($\eps$-approximate BNE) of FPA is defined as follow.
\begin{definition}[$\eps$-approximate Bayesian Nash equilibrium]
Let $n, m \geq 2$. Given a first-price auction ($\mN, \mB, \D,\Gamma$), a strategy profile $s = (s_1, \ldots, s_{n})$ is an $\eps$-approximate Bayesian Nash equilibrium ($\eps$-approximate BNE) if for any player $i\in [n]$, we have
\begin{align*}
    \E_{v \sim \D}\big[u_i(v_i; s_i(v_i), s_{-i}(v_{-i}))\big] \geq \E_{v \sim \D}\big[u_i(v_i; \bs(v_i, s_{-i}), s_{-i}(v_{-i}))\big] - \eps, 
\end{align*}
where $\bs(v_i, s_{-i}) \in \mB$ is the best response of player $i$ given other players' strategy $s_{-i}$, i.e.
\begin{align*}
\bs(v_i, s_{-i})\in \arg\max_{b\in \mB} \E_{v_{-i}\sim \D_{-i}}\big[u_i(v_i; b, s_{-i}(v_{-i}))\big].
\end{align*}
\end{definition}

The existence and the PPAD membership of finding a $1/\poly(n)$-approximate BNE can be established via a similar approach of \cite{filos2021complexity} (In particular, Theorem 4.1 and Theorem 4.4 of \cite{filos2021complexity}), and we omit the standard proof here.

We shall also use another notion of equilibrium which is more convenient in our hardness reduction. The $\eps$-approximately well-supported Bayesian Nash equilibrium ($\eps$-BNE) is defined as\footnote{We note the $\eps$-approximate BNE is known also {\em ex-ante} approximate BNE, and the $\eps$-BNE is known as {\em ex-interim} approximate BNE in some of the literature.}
\begin{definition}[$\eps$-approximately well-supported Bayesian Nash equilibrium]
Let $n, m \geq 2$. Given a first-price auction ($\mN, \mB, \D,\Gamma$), a strategy profile $s = (s_1, \ldots, s_{n})$ is an $\eps$-approximately well-supported Bayesian Nash equilibrium ($\eps$-BNE) if for any player $i\in [n]$ and $v_i \in [0, 1]$, we have
\begin{align*}
    \E_{v_{-i} \sim \D_{-i}}\big[u_i(v_i; s_i(v_i), s_{-i}(v_{-i}))\big] \geq \E_{v_{-i}\sim \D_{-i}}\big[u_i(v_i; \bs(v_i, s_{-i}), s_{-i}(v_{-i}))\big] - \eps. 
\end{align*}
\end{definition}

The notion of $\eps$-BNE and $\eps$-approximate BNE can be reduced to each other in polynomial time, losing at most a polynomial factor of precision. 
It is clear that an $\eps$-BNE is also an $\eps$-approximate BNE. Lemma \ref{lem:notion} states the other direction and the proof is deferred to the appendix.
\begin{lemma}\label{lem:notion}
Given a first-price auction ($\mN, \mB, \D,\Gamma$) and an $\eps$-approximate BNE $s$, there is a polynomial time algorithm that maps $s$ to an $\eps'$-BNE, where $\eps' = (2n+10)\sqrt{\eps}$. 
\end{lemma}


\section{PPAD-hardness}
\label{sec:ppad}


Recall our main hardness result

\Hardness*

In the rest of section, we construct the hard instances of FPA in Section \ref{sec:construction} and provide some basic properties in Section \ref{sec:basic}. 
We reduce from the $\eps$-generalized-circuit problem in Section \ref{sec:reduce}.

\subsection{Construction of first price auctions}
\label{sec:construction}

It suffices to prove finding $\eps$-BNE is hard for some $\eps = 1/\poly(n)$ due to Lemma \ref{lem:notion}.
We will use the following three parameters in the construction:
\begin{align*}
\eps = \frac{1}{n^{40}}, \quad \delta = \frac{1}{n^{10}} \quad \text{and} \quad \beta = \frac{1}{n^4}.
\end{align*}

We describe the bidding space $\mB$, the valuation distribution $\D$ and the tie-breaking rule $\Gamma$.

\paragraph{Bidding space.} The bidding space $\mB = \{b_0,b_1, b_2\}$ contains $3$ bids in total, where $b_0 = 0$, 
$$b_1 = \frac{\delta^2}{n^4}\quad\text{and}\quad b_2 = \frac{\delta}{n^2}.$$
\paragraph{Valuation distribution.} 
There are $ n+1$ players --- $n$ standard players indexed by $[n]$ and one pivot player $n+1$. We will describe  the value distribution $\D_i$ of player $i$ by specifying its density function $p_i: [0,1]\rightarrow \R^{+}$. 
The density function $p_{n+1}$ of the pivot player is set as follows:
\begin{align*}
    p_{n+1}(v) = 
    \begin{cases}
    1/(2\eps) & v \in [0,\eps]\\
    1/(2\eps) & v \in [1-\eps, 1]
    \end{cases}
   .
\end{align*}
In another word, $\D_{n+1}$ has $0.5$ probability mass around $0$ and $0.5$ probability mass around $1$. 

The density function $p_i$ of each standard player $i \in [n]$ is set as follows:
\begin{align*}
    p_i(v) =  
    \begin{cases}
    (1 - (2 + \beta)\delta)/\eps & v \in [0,\eps]\\
    \delta /\eps & v \in [b_2 - \eps, b_2]\\
    \tilde{p}_i(v) & v \in (b_2, 1 - \eps)\\
    \delta/\eps & v \in [1- \eps, 1] 
    \end{cases}
\end{align*} 
where $\tilde{p_i}(v)$ is defined over $(b_2,1-\eps)$, satisfies $\int_{b_2}^{1-\eps}\tilde{p}_i(v) \mathsf{d} v = \beta \delta$, but will be specified later in the reduction in Section \ref{sec:reduce}.
In short, a standard player $i$ has most its probability mass around $0$, $\delta$ mass around $b_2$, $\delta$ mass around $1$ and $\beta\delta$ mass in $(b_{2}, 1-\eps)$
  to be specified later.

\paragraph{Tie-breaking rule.}
We describe the trilateral tie-breaking rule $\Gamma$ as follows. For any bidding profile $\beta$ with $2\le |W(\beta)|\le 3$, the tie-breaking rule depends on the presence of $n+1$ in $W(\beta)$:
\begin{itemize}
\item Suppose $n+1 \notin W(\beta)$. Then
\begin{flushleft}\begin{itemize}
\item If $|W(\beta)| = 2$, i.e., $W(b) = \{i, j\}$, the tie-breaking rule is given by a matrix $\Sigma^{A} \in [0,1]^{n\times n}$ such that player $i$ obtains $\Sigma^{A}_{i, j}$ unit of the item and player $j$ obtains 
$\Sigma^{A}_{j,i}$ unit. 
So the matrix $\Sigma^A$ needs to satisfy $\Sigma^A + (\Sigma^A)^{\top} = (J_n - I_n)$. 
We will specify $\Sigma^A$ in the reduction later but will  guarantee that all of its off-diagonal entries lie in $[1/4,3/4]$.
\item If $|W(\beta)| = 3$, then we use the uniform allocation.
\end{itemize}\end{flushleft}
\item Suppose  $n+1 \in W(\beta)$. Then
\begin{flushleft}\begin{itemize}
\item If $|W(\beta)| = 2$, then the item is fully allocated to the pivot player $n+1$.
\item If $|W(\beta)| = 3$, i.e., $W(b) = \{i, j, n+1\}$, then the tie breaking is given by a matrix $\Sigma_B \in [0,1]^{n\times n}$  such that player $i$ obtains $\Sigma^{B}_{ i, j}$ unit of the item, player $j$ obtains $\Sigma^B_{j,i}$ unit and player $n+1$ obtains $1-\Sigma^B_{i,j}-\Sigma^B_{j,i}$ unit. So the matrix $\Sigma^B$ needs to satisfy $\Sigma^B + (\Sigma^B)^{\top} \leq J_n - I_n$, i.e., $\Sigma^B + (\Sigma^B)^{\top}$ is entrywise dominated by $J_n - I_n$.
\end{itemize}\end{flushleft}
\end{itemize}

\subsection{Basic properties}\label{sec:basic}

Let $s=(s_1,\ldots,s_{n+1})$ be an $\eps$-BNE of the instance.
We prove a few properties of $s$ in this subsection.
Given  $s$, for each player $i$ we define $f_i: \mB\rightarrow [0,1]$ and $F_i: \mB\rightarrow [0,1]$ as follows:
\[
f_i(b ) = \Pr_{v_i\sim \D_i}[s_i(v_i) = b ] \quad \text{and} \quad F_i(b ) = \Pr_{v_i\sim \D_i}[s_i(v_i) \leq b ].
\]
In the rest part of section, we abbreviate 
\[
\Gamma_i(b, s_{-i}) := \E_{v_{-i}\sim \D_{-i}}[\Gamma_i(b, s_{-i}(v_{-i}))] \quad \text{and}\quad u_i(v_i; b, s_{-i}) := \E_{v_{-i}\sim \D_{-i}}[u_i(v_i; b; s_{-i}(v_{-i}))] 
\]
when there is no confusion.

We start with the following lemma.
\begin{lemma}[Separable bid]
\label{lem:separable}
In any $\eps$-BNE, the equilibrium strategy $s$ satisfies
\begin{itemize}
\item For a standard player $i \in [n]$, its equilibrium strategy satisfies 
\begin{itemize}
\item when $v_i \in [0, \eps]$, $s_i(v_i) = b_0$; 
\item when $v_i \in [b_2 - \eps, b_2]$, $s_i(v_i) = b_1$; and 
\item when $v_i \in [1-\eps, 1]$, $s_i(v_i)= b_2$. 
\end{itemize}
\item For the pivot player, its equilibrium strategy satisfies 
\begin{itemize}
\item when $v_{n+1}\in [0, \eps]$, $s_{n+1}(v_{n+1}) = b_0$; and 
\item when $v_{n+1}\in [1-\eps, 1]$, $s_{n+1}(v_{n+1}) = b_2$.
\end{itemize}
\end{itemize}
\end{lemma}
\begin{proof}
The claim of $s_i(\eps) = 0$ holds trivially for all $i \in [n+1]$ due to the no-overbidding assumption.
A standard player $i$ chooses between $b_0$ and $b_1$ for $v_i \in [b_2 - \eps, b_2]$.
The allocation probability $\Gamma_i(b_0, s_{-i})$ of bidding $b_0$ satisfies $\Gamma_i(b_0, s_{-i}) \leq \frac{1}{n+1}$, hence the utility of bidding $b_0 = 0$ is at most 
\[
u_i(v_i; b_0, s_{-i}) = (v_i - b_0) \cdot \Gamma_i(b_0, s_{-i}) \leq \frac{1}{n} b_2.
\]
The allocation probability of bidding $b_1$ is at least 
\[
\Gamma_i(b_1, s_{-i}) \geq \prod_{i\in [n+1]}f_i(b_0) \geq (1 - (2+\beta)\delta)^{n} \cdot \frac{1}{2} \geq \frac{1}{3}
\] 
hence the utility of bidding $b_1$ is at least  
\[
u_i(v_i; b_1, s_{-i}) = (v_i - b_1) \cdot \Gamma_i(b_1, s_{-i}) \geq  (b_2 -\eps - b_1) \cdot \frac{1}{3} > u_i(v; b_0, s_{-i})  + \eps.
\]

Finally, we analyse the equilibrium strategy around $v \in [1-\eps, 1]$ for all $n+1$ players. Via an analysis similar to the above argument, it is clear that both standard players and the pivot player would choose between $b_1$ and $b_2$.
For a standard player $i \in [n]$,  
the allocation probability of bidding $b_{2}$ satisfies
\begin{align}
    \Gamma_i(b_{2}, s_{-i}) \geq &~ F_{n+1}(b_1) \cdot \prod_{j \in [n]\setminus [i]}F_{j}(b_1) \notag\\
    \geq &~ F_{n+1}(b_1) \cdot \left(\prod_{j \in [n]\setminus [i]}f_{j}(b_0) + \sum_{j \in [n]\setminus [i]}f_j(b_1) \prod_{r\in [n]\setminus \{i, j\}}f_r(b_0) \right)\notag \\
    \geq &~ \frac{1}{2}  \left(\prod_{j \in [n]\setminus [i]}f_j(b_0) + \sum_{j \in [n]\setminus [i]}f_j(b_1) \prod_{r\in [n]\setminus \{i, j\}}f_r(b_0) \right) + \frac{1}{2}f_{n+1}(b_1) ,\label{eq:separable-4}
\end{align}
where the last step holds as $f_{n+1}(b_0) = \frac{1}{2}$ and
\begin{align}
\prod_{j \in [n]\setminus [i]}f_j(b_0)  \geq (1 -(2+\beta)\delta)^{n-1} \geq \frac{1}{2} . \label{eq:separable-5}
\end{align}
The allocation probability of bidding $b_1$ satisfies
\begin{align}
    \Gamma_i(b_{1}, s_{-i}) \leq &~  f_{n+1}(b_0) \cdot \left(\prod_{j \in [n]\setminus [i]}f_j(b_0) + \sum_{j \in [n]\setminus [i]}f_j(b_1) \cdot \Sigma^A_{i, j} \prod_{r\in [n]\setminus \{i, j\}}f_r(b_0) + 9n^2\delta^2\right) \notag\\
    +&~  f_{n+1}(b_1) \cdot \sum_{j \in [n]\setminus\{i\}}f_j(b_1) \notag \\
     \le &~ \frac{1}{2} \left(\prod_{j \in [n]\setminus [i]}f_j(b_0) + \sum_{j \in [n]\setminus [i]}f_j(b_1) \cdot \Sigma^{A}_{i, j} \prod_{r\in [n]\setminus \{i, j\}}f_r(b_0)\right) + f_{n+1}(b_1)\cdot  3n\delta + 9n^2\delta^2\label{eq:separable-3}.
\end{align}
Here the first step holds since (1) when the pivot player bids $b_0$, player $i$ obtains $\Sigma^{A}_{i, j}$ unit of item when (only) player $j$ bids $b_1$, and the probability of at least two players bidding $b_1$ is bounded as
\begin{align*}
\sum_{j, r \in [n]\backslash \{i\}} \Pr[s_r(v_r) = b_1 \wedge s_j(v_j) = b_1] \leq n^2\cdot  9\delta^2,
\end{align*}
(2) when the pivot player bids $b_1$, the player $i$ obtains the item only if there exists at least one other standard player $j$ bids $b_1$ as the tie breaking rule assigns the item fully to player $n+1$ when there are only two winners. The second step follows from $f_j(b_1)\leq 3\delta$
and that the pivot player does not bid $b_0$ in $[1-\eps,1]$ so $f_{n+1}(b_0)=1/2$.

Subtracting Eq.~\eqref{eq:separable-3} and Eq.~\eqref{eq:separable-4}, one obtains
\begin{align}
\Gamma_i(b_{2}, s{-i})  - \Gamma_i(b_{1}, s_{-i})\notag \\
\geq &~   \frac{1}{2}\sum_{j \in [n]\setminus [i]}f_j(b_1) \cdot (1 - \Sigma^{A}_{i, j}) \prod_{r\in [n]\setminus \{i, j\}}f_r(b_0) +\frac{1}{2}f_{n+1}(b_1) -3n\delta f_{n+1}(b_1)-9n^2\delta^2 \notag \\
\geq &~ \frac{1}{2} \cdot (n-1) \cdot \delta \cdot \frac{1}{4} \cdot \frac{1}{2} + \frac{1}{2}f_{n+1}(b_1) -3n\delta f_{n+1}(b_1)-9n^2\delta^2\geq \frac{n\delta}{32}\label{eq:diff}.
\end{align} 
The second step holds due to $f_j(b_1) \geq \delta$, $\Sigma^{A}_{i, j} \in [1/4, 3/4]$ and Eq~\eqref{eq:separable-5}.
Hence we claim player $i$ prefers $b_2$ than $b_1$ at value $v_{i} \in [1-\eps, 1]$, since
\begin{align*}
u_{i}(v_{i}; b_2, s_{-i}) - u_{i}(v_i; b_1, s_{-i}) =&~  (v_{i} - b_2) \cdot \Gamma_{i}(b_2, s_{-i}) - (v_{i} - b_1) \cdot \Gamma_{i}(b_1, s_{-i}) \\
\geq &~ v_{i} \cdot (\Gamma_{i}(b_2, s_{-i}) - \Gamma_{i}(b_1, s_{-i})) - b_2\\
\geq &~ (1-\eps) \cdot \frac{n\delta}{32} - b_2 > \eps.
\end{align*}

Finally, for the pivot player $n+1$, the allocation probability of bidding $b_1$ satisfies
\begin{align}
\Gamma_{n+1}(b_1, s_{-i}) \leq &~ \prod_{i \in [n]}F_i(b_1)
\end{align}
and the allocation probability of $b_2$ satisfies
\begin{align}
\Gamma_{n+1}(b_2, s_{-i}) \geq &~ \prod_{i \in [n]}F_i(b_1) + \sum_{i \in [n]}f_i(b_2) \prod_{j \in [n]\backslash \{i\}}F_i(b_1)\notag \\
\geq &~\prod_{i \in [n]}F_i(b_1) + \frac{n\delta}{2}.\label{eq:separable2}
\end{align}
The first step holds since the tie-breaking rule favors player $n+1$ when at most one player in $[n]$ bids $b_2$, the second step holds due to $f_i(b_2) \geq \delta$ and Eq.~\eqref{eq:separable-5}.
Hence, at any $v_{n+1} \in [1-\eps, 1]$ we have
\begin{align*}
u_{n+1}(v_{n+1}; b_2, s_{-i}) - u_{n+1}(v; b_1, s_{-i})  = &~ (v_{n+1} - b_2) \cdot \Gamma_{n+1}(b_2, s_{-i}) - (v_{n+1} - b_1) \cdot \Gamma_{n+1}(b_1, s_{-i}) \\
\geq &~ v_{n+1} \cdot (\Gamma_{n+1}(b_2, s_{-i}) - \Gamma_{n+1}(b_1, s_{-i})) - b_2\\
\geq &~ (1-\eps) \cdot \frac{n\delta}{2} - b_2 > \eps.
\end{align*}
We conclude the proof of the lemma here.
\end{proof}

Lemma \ref{lem:separable} confirms that $b_0, b_1, b_2$ would appear in an $\eps$-BNE profile for every player $i\in [n]$. It still remains to determine at which value point a standard player $i \in [n]$ jumps from $b_1$ to $b_2$ in $s_i$.
Let $\tau_{i} \in (b_2, 1-\eps)$ be the jumping point from $b_{1}$ to $b_{2}$ of a standard player $i$. The following formula is convenient to use.

\begin{lemma}[Jumping point formula]
\label{lem:formula}
The jumping point $\tau_i$ of a standard player $i \in [n]$ satisfies
\begin{align*}
    \tau_{i} = b_{2} + \frac{\Gamma_i(b_{1}, s_{-i})\cdot (b_{2} - b_1)}{\Gamma_i(b_{2}, s_{-i}) - \Gamma_i(b_{1}, s_{-i})} \pm O\left(\frac{\eps}{\delta}\right).
\end{align*}
\end{lemma}
\begin{proof}
At any value point $v \in [0, 1]$, recall the utility of bidding $b_1$ equals
\begin{align*}
    u_{i}(v_i; b_1,s_{-i}) = (v_i - b_1)\Gamma_i(b_1, s_{-i})
\end{align*}
and the utility of bidding $b_{2}$ equals 
\[
u_{i}(v_i; b_{2},s_{-i}) = (v_i - b_{2})\Gamma_i(b_{2}, s_{-i}).
\]
Solving for $u_{i}(\tau_i, b_1,s_{-i}) = u_{i}(\tau_i, b_{2},s_{-i}) \pm \eps$, one obtains
\begin{align*}
    \tau_{i} = &~  \frac{\Gamma_i(b_{2}, s_{-i}) b_{2} - \Gamma_i(b_{1}, s_{-i})b_1 \pm \eps }{\Gamma_i(b_{2}, s_{-i}) - \Gamma_i(b_{1}, s_{-i})} \\
    = &~ b_{2} + \frac{\Gamma_i(b_{1}, s_{-i})\cdot (b_{2} - b_1) \pm \eps}{\Gamma_i(b_{2}, s_{-i}) - \Gamma_i(b_{1}, s_{-i})}\\ 
    = &~ b_{2} + \frac{\Gamma_i(b_{1}, s_{-i})\cdot (b_{2} - b_1)}{\Gamma_i(b_{2}, s_{-i}) - \Gamma_i(b_{1}, s_{-i})} \pm O\left(\frac{\eps}{\delta}\right) \label{eq:diff}.
\end{align*}
The last step follows from Eq.~\eqref{eq:diff}, and this finishes the proof of the lemma.
\end{proof}

Let $x_i \in [0, \beta\delta]$ be the probability mass over the interval $(b_2, \tau_i)$, i.e., $x_i = \int_{b_2}^{\tau_i}p_i(v)\mathsf{d}v$, which we will refer to the jumping probability of $s_i$. We state a few facts that will be used repeatedly.
\begin{lemma}[Basic facts]
\label{lem:basic-fact}\ 
\begin{itemize}
\item For any standard player $i \in [n]$, we have
\[
\prod_{j \in [n]\setminus \{i\}}f_j(b_0) = 1 -(n-1)(2+\beta)\delta \pm O(n^2\delta^2)
\]
and
\[
\prod_{j \in [n]\setminus \{i\}}F_j(b_1) = 1 -(n-1)(\beta+1)\delta + \sum_{j \in [n]\setminus \{i\}}x_j \pm O(n^2\delta^2).
\]
\item For any $e \in \{1,2\}$, we have 
$$\sum_{i\neq j \in [n]} \Pr\big[s_i(v_i) = b_e \wedge s_j(v_j) = b_e\big] = O(n^2\delta^2).$$
\end{itemize}
\end{lemma}
\begin{proof}
For the first claim, we have
\begin{align*}
\prod_{j \in [n]\setminus \{i\}}f_j(b_0) = (1 - (2+\beta)\delta)^{n -1} = 1 - (n-1)(2+\beta)\delta  \pm O(n^2\delta^2)
\end{align*}
due to the choice of $\delta$. Similarly we have (using $x_i\in [0,\beta\delta]$
\begin{align*}
\prod_{j \in [n]\setminus \{i\}}F_j(b_1)  =  &~ \prod_{j \in [n]\setminus \{i\}}(1 - (1+\beta)\delta + x_i) \\
= &~ 1 - (n-1)(1+\beta)\delta + \sum_{j \in [n]\backslash \{i\}}x_j \pm O(n^2\delta^2).
\end{align*}
For the second claim, for any $e\in \{1,2\}$, we have
\begin{align*}
\sum_{i \neq j \in [n]} \Pr\big[s_i(v_i) = b_e \wedge s_j(v_j) = b_e\big] \leq \sum_{i\neq j \in [n]}(1+\beta)\delta^2= O(n^2\delta^2).
\end{align*}
We conclude the proof here.
\end{proof}

The key step is to determine the jumping point, where we use approximation.
\begin{lemma}[Jumping point]
\label{lem:jump}
The jumping point $\tau_i$ of a standard player $i\in [n]$ satisfies
\begin{align*}
\tau_i = \left(\frac{\Delta_{i,1} - \Delta_{i,2}}{\Delta_{i,1}^2} \pm O\left(\frac{\beta^2}{\Delta_{i,1}}\right) \right)\cdot b_2 
\end{align*}
where
\begin{align*}
\Delta_{i, 1} := &~ (n-1)\delta + \sum_{j\in [n]\setminus \{i\}}\Big(\beta\delta\cdot \Sigma^{A}_{i,j} + (1+\beta)\delta\cdot \Sigma^{B}_{i,j}\Big) \in \big[(n-1)\delta, 2n\delta\big]\quad\text{and}\\
\Delta_{i, 2} := &~ \sum_{j \in [n]\backslash \{i\}} \Big(1 - 2\Sigma^{A}_{ i, j} - \Sigma^{B}_{i, j}\Big)x_{j}  \in \big[-2n\beta\delta, n\beta\delta\big].
\end{align*}
\end{lemma}
\begin{proof}
For any standard player $i \in [n]$, we compute when it jumps from $b_{1}$ to $b_2$ using the formula in Lemma \ref{lem:formula}.
To do so, we first compute $\Gamma_i(b_{1}, s_{-i})$ and $\Gamma_i(b_{2}, s_{-i})$.
\begin{align}
\Gamma_i(b_{1}, s_{-i}) =&~ f_{n+1}(b_0) \cdot \left( \prod_{j \in [n]\setminus \{i\}}f_j(b_0) + \sum_{j\in [n]\setminus\{i\}}f_j(b_1)\prod_{j \in [n]\setminus \{i, j\}}f_j(b_0)\cdot \Sigma^{A}_{ i, j} \pm O(n^2\delta^2)  \right)\notag\\
= &~ \frac{1}{2}  \left( 1 - (n-1)(2+\beta)\delta + \sum_{j\in [n]\setminus\{i\}} (\delta + x_j)\Sigma^{A}_{i,j}  \right)\pm O(n^2\delta^2). \label{eq:prob1}
\end{align}
Here the first step follows from the tie-breaking rule and the second claim of Lemma \ref{lem:basic-fact}, the second step follows from $f_j(b_1) = \delta + x_j$ and the first claim of Lemma \ref{lem:basic-fact}.

The allocation probability of bidding $b_2$ obeys
\begin{align*}
    \Gamma_i(b_{2}, s_{-i}) = &~ f_{n+1}(b_0) \cdot \left(\prod_{j \in [n]\setminus \{i\}}F_j(b_1) +  \sum_{j\in [n]\setminus\{i\}}f_j(b_2) \prod_{j \in [n]\setminus \{i, j\}}F_j(b_1) \cdot \Sigma^{A}_{i, j} \pm O(n^2\delta^2) \right)\\
    &~ + f_{n+1}(b_2) \cdot \left(\sum_{j\in [n]\setminus\{i\}}f_j(b_2)\prod_{k \in [n]\setminus \{i, j\}}F_k(b_1) \cdot \Sigma^{B}_{i, j} \pm O(n^2\delta^2) \right)\\
    = &~ \frac{1}{2} \left(1 - (n-1)(1+\beta)\delta + \sum_{j \in [n]\setminus \{i\}}x_j + \sum_{j \in [n]\setminus \{i\}}(\delta + \beta\delta - x_j)\Sigma^{A}_{i,j}  \right)\\
    &~ + \frac{1}{2}\sum_{j \in [n]\setminus \{i\}}(\delta + \beta\delta - x_j)\Sigma^{B}_{i, j} \pm O(n^2 \delta^2).
\end{align*}
The first step uses the tie breaking rule and requires some explanations. In particular, (1) when the pivot player $n+1$ bids $b_0$, the player $i$ obtains $1$ unit of item when other players bid less than $b_2$, $\Sigma_{A, i, j}$ unit of item when only player $j$ bids $b_2$; we also make use of the second claim of Lemma \ref{lem:basic-fact} to omit the other case; (2) when the player $n+1$ bids $b_2$, the player $i$ obtains $0$ unit of good when no other players bid $b_0$ and obtains $\Sigma_{B,i, j}$ unit of goods when one other player $j$ bids $b_2$, and we omit other cases using Lemma \ref{lem:basic-fact}.
The second step follows from Lemma \ref{lem:basic-fact} and $f_j(b_2) = \delta + \beta\delta - x_j$.

Combining the above expression, we have
\begin{align}
\Gamma_i(b_{2}, s_{-i}) - \Gamma_i(b_{1}, s_{-i}) = &~ \frac{1}{2}(n-1)\delta + \frac{1}{2}\sum_{j\in [n]\setminus \{i\}}\Big(\beta\delta\cdot \Sigma^{A}_{i,j} + (1+\beta)\delta\cdot \Sigma^{B}_{i,j}\Big) \notag \\
&~ + \frac{1}{2}\sum_{j \in [n]\backslash \{i\}} \Big(1 - 2\Sigma^{A}_{i, j} - \Sigma^{B}_{i, j}\Big)x_{j} \pm O(n^2\delta^2)\label{eq:diff2}.
\end{align}
Let $\Delta_{i,1}$ and $\Delta_{i,2}$ be defined as in the statement of the lemma.
Note that $\Delta_{i, 1}$ does not depend on $\{x_j\}_{j\ne i}$ while $\Delta_{i, 2}$ depends on $\{x_j\}_{j \neq i}$.
It is easy to see that 
\begin{align}
\Delta_{i,1} \in \big[(n-1)\delta, 2n\delta\big] \quad \text{and} \quad  \Delta_{i, 2} = \big[-2(n-1)\beta\delta, (n-1)\beta\delta\big]\label{eq:range}.
\end{align}

Finally we can compute the jumping point $\tau_i$ using Lemma \ref{lem:formula} as follows:
\begin{align*}
\tau_i = &~  b_2 + \frac{\Gamma_i(b_{1}, s_{-i})\cdot (b_{2} - b_1)}{\Gamma_i(b_{2}, s_{-i}) - \Gamma_i(b_{1}, s_{-i})} \pm O\left(\frac{\eps}{\delta}\right) \\
= &~ b_2 + \frac{1 \pm O(n\delta)}{\Delta_{i,1}+\Delta_{i,2}} \cdot b_2 \pm O\left(\frac{\eps}{\delta}\right)\\
= &~ \left(\frac{\Delta_{i,1} - \Delta_{i,2}}{\Delta_{i,1}^2} \pm O\left(\frac{\beta^2}{\Delta_{i,1}}\right) \right)\cdot b_2 
\end{align*}
The second step follows from Eq.~\eqref{eq:diff2}, $\Gamma_i(b_1, \beta_i) = (1/2) \pm O(n\delta)$ (see Eq.~\eqref{eq:prob1}) and the choice of $b_1,b_2$.
The last step follows from Eq.~\eqref{eq:range}.
\end{proof}

\subsection{Reduction from generalized circuit}
\label{sec:reduce}

Given $\alpha<\beta$,
  we write $\mathsf{T}_{[\alpha, \beta]}: \R \rightarrow [\alpha, \beta]$ to denote the  truncation function with  $$\mathsf{T}_{[\alpha, \beta]}(x) = \min\big\{\max\{x, \alpha\}, \beta\big\}.$$ We recall the  generalized circuit problem  \cite{chen2009settling} and  present a simplified version from \cite{filos2021complexity}. 
\begin{definition}[(Simplified) generalized circuit]
	\label{def:generelized-circuit}
	A generalized circuit is a tuple $(V, G)$, where $V$ is a set of nodes and $G$ is a collection of gates.
    Each node $v \in V$ is associated with a gate $G_v$ that falls into one of two types $\{G_{1-}, G_{+}\}$: If $G_v$ is a $G_+$ gate, then it has two input nodes $v_1,v_2\in V\setminus \{v\}$; if it is a $G_{1-}$ gate then it takes one input node $v_1\in V\setminus \{v\}$.
Given $\kappa>0$,
a $\kappa$-approximation solution to $(V,G)$ is an assignment $x\in [0, 1]^{V}$such that for every node $v$:
     \begin{itemize}
    \item If $G_v$ is a $G_{+}$ gate and  takes input nodes $v_1, v_2 \in V\backslash \{v\}$, then $x_v = \mathsf{T}_{[0, 1]}(x_{v_1} + x_{v_2}\pm \kappa) $
    \item If $G_v$ is a  $G_{1-}$ gate and takes  an input node $v_1 \in V\backslash \{v\}$, then   $x_{v} = \mathsf{T}_{[0,1]}(1 - x_{v_1}\pm \kappa) $.
    \end{itemize}   

\end{definition}

The generalized circuit problem is known to be PPAD-hard for constant $\kappa$.
\begin{theorem}[\cite{rubinstein2015inapproximability,deligkas2022pure}]
There is a constant $\kappa>0$ such that it is PPAD-hard to find an $\kappa$-approximate solution of a generalized circuit.
\end{theorem}

We prove Theorem \ref{thm:hardness} via a reduction from the generalized circuit problem.

Given an instance of generalized circuit defined over nodes set $V$ ($|V| = m$), we let $V_1 = [m_1]$ be the set of nodes with gate $G_{+}$  and $V_2 = [m_1+1: m]$  be the set of nodes with gate $G_{1-}$.
We construct an instance of first price auction with $n = m_1 + 2(m - m_1) = 2m - m_1$ standard players and one pivot player.

Let $\mN = \mN_1 \cup \mN_{2} \cup \mN_{3}$ be the set of standard players, where $ \mN_1 = [m_1]$, $ \mN_{2} = [m_1+1: m]$ and $\mN_{3} = [m+1: 2m - m_1]$.
From a high level, we use players in $\mN_1$ to represent the set of nodes with $G_{+}$ gates, players in $\mN_{2}$ to represent the set of nodes with $G_{1-}$ gates. 
Players in $\mN_{3}$ are used in constructing $G_{1-}$. 
We first specify the probability density $\tilde{p}$ over interval $(b_2, 1-\eps)$ and the tie-breaking matrices $\Sigma^A$ and $\Sigma^B$ to complete the description of the FPA instance.

\begin{itemize}
\item For player $i$ in $\mN_1$ (i.e., $i \in [m_1]$), its valuation distribution $\tilde{p}_i$ is uniform over the interval
\[
\left[\frac{1}{\Delta_{i, 1}} \cdot b_2, \left(\frac{1}{\Delta_{i, 1}} + \frac{1}{10}\cdot\frac{\beta\delta}{\Delta_{i, 1}^2}\right)\cdot b_2\right]
\]
with a total probability mass of $\beta\delta$. Let $i(1), i(2) \in [m]=\mN_1\cup \mN_2$ be the input nodes of the $G_{+}$ gates, we set $\Sigma^{B}_{i, i(1)} = \Sigma^{B}_{i, i(2)}= 1/10$.
\item For player $m_1 + j \in \mN_{2}$ (i.e., $j \in [m - m_1]$), its valuation distribution $\tilde{p}_{m_1+j}$ is uniform over 
\[
\left[\left(\frac{1}{\Delta_{m_1 + j, 1}} - \frac{1}{10}\cdot\frac{\beta\delta}{\Delta_{m_1 +j, 1}^2}\right)\cdot b_2, \frac{1}{\Delta_{m_1+j, 1}}\cdot b_2\right]
\] 
with a total probability mass of $\beta\delta$.
We set $\Sigma^{A}_{m_1 + j, m+ j} = 9/20.$
\item For player $m + j \in \mN_{3}$ (i.e., $j \in [m - m_1]$), its valuation distribution $\tilde{p}_{m+j}$ is uniform over  
\[
\left[\left(\frac{1}{\Delta_{m+j, 1}} + \frac{1}{10}\cdot\frac{\beta\delta}{\Delta_{m+j, 1}^2}\right)\cdot b_2, \left(\frac{1}{\Delta_{m+j, 1}} + \frac{1}{5}\cdot\frac{\beta\delta}{\Delta_{m+j, 1}^2}\right)\cdot b_2\right]
\]
with a total probability mass of $\beta\delta$. We set $\Sigma^{A}_{m + j, m_1+ j} = 11/20$. Let $j(1) \in [m]$ be the input node of $G_{1-}$, then set $\Sigma^{B}_{m+j, j(1)} =  {1}/{5}$.
\item For any entry of $\Sigma^{A}$ that has not been determined above, we set it to be $1/2$, and for any entry of $\Sigma^{B}$ that has not been determined, we set it to be $0$.
\end{itemize}

It is easy to verify that $\Sigma^{A}$ and $\Sigma^B$ satisfy the following properties as promised earlier: (1) the off-diagonal entries of $\Sigma^{A}$ lie in $[1/4, 3/4]$;  (2) $\Sigma^{A} + (\Sigma^{A})^{\top} = J_n - I_n$; and (3) the off-diagonal entries of $\Sigma^{B}$ belong to $[0,1/2]$.

Letting $\kappa: = n\beta = {1}/{n^3}$, we prove that any $\eps$-BNE of the first price auction gives an $O(\kappa)$-approximate solution to the generalized circuit.
Indeed the following lemma shows that by taking $x_{i}' =  {x_{i}}/{(\beta\delta)}$, where $x_i$ is the jumping probability of $s_i$, we obtain an $O(\kappa)$-approximate solution $(x_1',\ldots,x_m')$ to the input generalized circuit. This finishes the proof of Theorem \ref{thm:hardness}.
\begin{lemma}
Given an $\eps$-BNE of the first price auction and let $(x_1, \ldots, x_n)\in [0, \beta\delta]^{n}$ be the tuple of jumping probabilities, then we have
\begin{itemize}
\item For any $i \in [m_1]$, $x_{i} = \mathsf{T}_{[0, \beta\delta]}(x_{i(1)} + x_{i(2)} \pm O(\kappa\beta\delta))$
\item For any $j \in [m- m_1]$, $x_{m_1 + j} = \mathsf{T}_{[0, \beta\delta]}(\beta\delta - x_{j(1)} \pm O(\kappa\beta\delta))$
\end{itemize}
\end{lemma}
\begin{proof}
For the first claim, for any $i \in [m_1]$, one has
\begin{align*}
\Delta_{i, 2} = \sum_{r \in [n]\backslash \{i\}} \Big(1 - 2\Sigma^{A}_{i, r} - \Sigma^{B}_{ i, r}\Big)x_{r} = -\frac{1}{10}x_{i(1)}  - \frac{1}{10}x_{i(2)},
\end{align*}
where the second step follows from $\Sigma^{A}_{i, r} = 1/2$ for all $r\in [n]\backslash \{i\}$, $\Sigma^{B}_{i, i(1)} = \Sigma^{B}_{i, i(2)} = 1/10$ and $\Sigma^{B}_{i, r} = 0$ for all other $r \in [n]\backslash \{i, i_1, i_2\}$.
By Lemma \ref{lem:jump}, one has
\begin{align*}
\tau_i = \left(\frac{\Delta_{i,1} - \Delta_{i,2}}{\Delta_{i,1}^2} \pm O\left(\frac{\beta^2}{\Delta_{i,1}}\right)\right) \cdot b_2 = \left(\frac{1}{\Delta_{i,1}} + \frac{1}{10}\cdot \frac{x_{i(1)} + x_{i(2)}}{\Delta_{i,1}^2}  \pm O\left(\frac{\beta^2}{\Delta_{i,1}}\right)\right) \cdot b_2.
\end{align*}
Since $\D_{i}$ is uniform over $[\frac{1}{\Delta_{i, 1}} \cdot b_2, (\frac{1}{\Delta_{i, 1}} + \frac{1}{10}\cdot\frac{\beta\delta}{\Delta_{i, 1}^2})\cdot b_2]$ with probability mass $\beta\delta$, we have 
\[
x_i = \mathsf{T}_{[0, \beta\delta]}(x_{i(1)} + x_{i(2)} \pm O(\kappa\beta\delta)).
\]

For the second claim, we first analyse the jumping probability of player $m_1 + j$. We have
\begin{align*}
\Delta_{m_1 + j, 2} = \sum_{r \in [n]\backslash \{m_1+j\}} \Big(1 - 2\Sigma^{A}_{m_1+j, r} - \Sigma^{B}_{m_1 +j, r}\Big)x_{r} = \frac{1}{10}x_{m+ j}  ,
\end{align*}
where the second follows from $\Sigma^{A}_{m_1 + j, m+j} = 9/{20}$, $\Sigma^{A}_{m_1+j, r} = 1/2$ for all $r\in [n]\backslash \{m_1 + j,m+j \}$, and $\Sigma^{B}_{m_1+j, r} = 0$ for all $r \in [n]\backslash \{m_1 +j \}$. Hence, by Lemma \ref{lem:jump}, we have
\begin{align*}
\tau_{m_1 + j} = &~ \left(\frac{\Delta_{m_1 +j,1} - \Delta_{m_1 + j,2}}{\Delta_{m_1 + j,1}^2} \pm O\left(\frac{\beta^2}{\Delta_{m_1 + j, 1}}\right)\right) \cdot b_2\\
=&~ \left(\frac{1}{\Delta_{m_1 +j,1}} - \frac{1}{10}\cdot \frac{x_{m+ j}}{\Delta_{m_1+j,1}^2}  \pm O\left(\frac{\beta^2}{\Delta_{m_1+j,1}}\right)\right) \cdot b_2
\end{align*}
Given that $\D_{m_1+j}$ is uniform over $[(\frac{1}{\Delta_{m_1 + j, 1}} - \frac{1}{10}\cdot\frac{\beta\delta}{\Delta_{m_1 +j, 1}^2})\cdot b_2, \frac{1}{\Delta_{m_1+j, 1}}\cdot b_2]
$ with mass $\beta\delta$, we have 
\begin{align}
\label{eq:fix1}
x_{m_1+j} = \mathsf{T}_{[0, \beta\delta]}(\beta\delta - x_{m+ j} \pm O(\kappa\beta\delta)) .
\end{align}
It remains to analyse the jumping probability of player $m+j$, and we have 
\begin{align*}
\Delta_{m + j, 2} = &~ \sum_{r \in [n]\backslash \{m+j\}} \Big(1 - 2\Sigma^{A}_{m+j, r} - \Sigma^{B}_{m+j, r}\Big)x_{r}\\
= &~ -\frac{1}{10}x_{m_1+ j}   - \frac{1}{5}x_{j(1)} = -\frac{1}{10}\left(\beta\delta - x_{m + j} + 2x_{j(1)}\right) \pm O(\kappa\beta\delta).
\end{align*}
Here the second step follows from $\Sigma^{A}_{m + j, m_1+ j} = 11/20$, $\Sigma^{A}_{m + j, r} = 1/2$ for $ r\in [n]\backslash \{m_1+j, m+j\}$, $\Sigma^{B}_{m+j, j(1)} = 1/5$ and $\Sigma^{B}_{m + j, r}= 0$ for any $r\in [n]\backslash \{m+j, j(1)\}$. The last step follows from Eq.~\eqref{eq:fix1}. 

Now, by Lemma \ref{lem:jump}, one has
\begin{align*}
\tau_{m + j} = &~ \left(\frac{\Delta_{m+j,1} - \Delta_{m+j,2}}{\Delta_{m+j,1}^2} \pm O\left(\frac{\beta^2}{\Delta_{m+j,1}}\right)\right) \cdot b_2 \\
= &~ \left(\frac{1}{\Delta_{m+j,1}} + \frac{1}{10}\cdot \frac{\beta\delta - x_{m+ j} + 2x_{j(1)} \pm O(\kappa\beta\delta)}{\Delta_{m+j,1}^2}  \pm O\left(\frac{\beta^2}{\Delta_{m+j,1}}\right)\right) \cdot b_2
\end{align*}
Given that $\D_{m+j}$ is uniform over $[(\frac{1}{\Delta_{m+j, 1}} + \frac{1}{10}\cdot\frac{\beta\delta}{\Delta_{m+j, 1}^2})\cdot b_2,(\frac{1}{\Delta_{m+j, 1}} + \frac{1}{5}\cdot\frac{\beta\delta}{\Delta_{m+j, 1}^2})\cdot b_2]$ with mass $\beta\delta$, we conclude that
\begin{align*}
x_{m+j} = \mathsf{T}_{[0, \beta\delta]}(x_{j(1)} \pm O(\kappa \beta\delta)).
\end{align*}
Plugging into Eq.~\eqref{eq:fix1}, we obtain
\begin{align*}
x_{m_1+j} = \mathsf{T}_{[0, \beta\delta]}(\beta\delta - x_{j(1)} \pm O(\kappa \beta\delta)).
\end{align*}
This completes the proof of the second claim.
\end{proof}

\section{PTAS}
\label{sec:ptas}

We present a PTAS for computing an $\eps$-approximate BNE in an FPA under the uniform tie-breaking rule.

\PTAS*

Our approach proceeds in the following four steps. 
Given an FPA $(\mN, \mB, \D, \Gamma)$ where $\Gamma$ is the uniform tie-breaking rule, we first round the bidding space $\mB$ and reduce its size to $O(1/\eps)$. 
We then prune the valuation distribution $\D$ and work on a weak notion of $(\eps, \delta)$-approximate BNE that relaxes the no-overbidding requirement.
In the third step, we argue the existence of an $(\eps, \delta)$-approximate BNE profile over a discretized space and in the final step, we develop a suitable searching algorithm for $(\eps, \delta)$-approximate BNE in the discretized space.

\vspace{+2mm}
{\noindent \bf Step 1: Rounding bids. \ \ }
Given a first-price auction with bidding space $\mB = \{b_0, b_1, \ldots, b_m\}$ with $0 = b_0 < b_1 < \cdots < b_m \leq 1$, 
we define $\mB_{\eps} = \{b_{0, \max}, b_{1, \max}, \ldots, b_{10\eps^{-1}, \max}\}$ as follows.
First we take $b_{0, \max} = b_0=0$.
Then for each $t \in [10\eps^{-1}]$, let 
$b_{t, \max}$ be the maximum bid in
$$
\mB_t := \mB \cap \left(\frac{(t-1)\eps}{10}, \frac{t\eps}{10}\right]
$$
if $\mB_t$ is not empty; and set
$b_{t,\max}=\nil$ (meaning that we don't add an element to $\mB_\eps$) if $\mB_t$ is empty. 
We prove that it suffices to find an $({\eps}/{2})$-approximate BNE over bidding space $\mB_{\eps}$.
\begin{lemma}
\label{lem:ptas-step1}
Given any first-price auction $(\mN, \mB, \D, \Gamma)$, let $B_{\eps}$ be the rounded bidding space defined above. Then any $ ({\eps}/{2})$-approximate BNE  of  $(\mN, \mB_{\eps}, \D, \Gamma)$ is also an $\eps$-approximate BNE of $(\mN, \mB, \D, \Gamma)$.
\end{lemma}
\begin{proof}
Let $s=(s_1,\ldots,s_n)$ be an $(\eps/2)$-approximate BNE of $(\mN, \mB_{\eps}, \D, \Gamma)$. By definition, each $s_i$ satisfies
\begin{align*}
\E_{v\sim \D}\big[u_i(v_i; s_i(v_i), s_{-i}(v_{-i})) - u_i(v_i; \bs_{\eps}(v_i, s_{-i}), s_{-i}(v_{-i}))\big] \geq -\frac{\eps}{2},
\end{align*}
where 
\[
\bs_{\eps}(v_i, s_{-i}) :=  \arg\max_{b\in \mB_{\eps}} \E_{v_{-i}\sim \D_{-i}}u_i(v_i; b, s_{-i}(v_{-i}))
\]
is the best response in $\mB_{\eps}$.
For each value $v_i \in [0,1]$, let 
\[
\eps_i(v_i) := \E_{v_{-i}\sim \D_{-i}}\big[u_i(v_i; s_i(v_i), s_{-i}(v_{-i})) - u_i(v_i; \bs_{\eps}(v_i, s_{-i}), s_{-i}(v_{-i}))\big]
\]
be the deficiency of utility at value point $v_i$ and we know that $\E_{v_i}[\eps_i(v_i)] \geq -{\eps}/{2}$. 

Let $\bs(v_i, s_{-i}) \in \mB$ be the best response in the original bidding space $\mB$.
Given that the case of $\bs(v_i, s_{-i})  = 0$ is trivial, we may assume without loss of generality that
  $$\bs(v_i, s_{-i}) \in \left(\frac{(t-1)\eps}{10}, \frac{t\eps}{10}\right],\quad\text{for some $t \in [10\eps^{-1}]$;}$$ hence $\mB_{t}\neq \emptyset$. 
Our goal is to prove that the revenue of bidding $b_{t,\max}$ or $0$ (both of which are in $\mB_\eps$) is at most $ {\eps}/{10}$ worse comparing to $\bs(v_{i}, s_{-i})$.

{\bf Case 1.}  Suppose $v_i \geq  {t\eps}/{10}$. Then it is valid to bid $b_{t, \max}$ and we have
\begin{align}
     \E_{v_{-i}\sim \D_{-i}}&\big[u_i(v_i; \bs_{\eps}(v_i, s_{-i}), s_{-i}(v_{-i})) - u_i(v_i; \bs(v_i, s_{-i}), s_{-i}(v_{-i}))\big]\notag \\
    \geq &~ \E_{v_{-i}\sim \D_{-i}}\big[u_i(v_i; b_{t, \max}, s_{-i}(v_{-i})) - u_i(v_i; \bs(v_i, s_{-i}), s_{-i}(v_{-i}))\big]\notag \\
    = &~ \E_{v_{-i}\sim \D_{-i}}[(v_i - b_{t, \max})\Gamma_i(b_{t, \max}, s_{-i}(v_{-i})) - (v_i - \bs(v_i, s_{-i}))\Gamma_i(\bs(v_i, s_{-i}), s_{-i}(v_{-i}))] \notag \\
    \geq &~ \bs(v_i, s_{-i}) - b_{t, \max}   \geq -\eps/10 \label{eq:step1-1}
\end{align}
where the first step follows from $b_{t, \max} \in \mB_{\eps}$, the third step follows from $v_i \geq b_{t, \max}$ and $$\Gamma_i(b_{t, \max}, s_{-i}(v_{-i})) \geq \Gamma_i(\bs(v_i, s_{-i}), s_{-i}(v_{-i})).$$ 

{\bf Case 2.}  Suppose $v_i \in (\frac{(t-1)\eps}{10}, \frac{t\eps}{10}]$ (note than we cannot have $v_i<(t-1)\eps/10$ due to the no overbidding assumption). In this case bidder $i$ may not be able to bid $b_{t, \max}$ due to the no-overbidding assumption. However, in this case the utility is small anyway so we can let bidder bid
  $b_0=0$ as follows:
\begin{align}
    \E_{v_{-i}\sim \D_{-i}}&\big[u_i(v_i; \bs_{\eps}(v_i, s_{-i}), s_{-i}(v_{-i})) - u_i(v_i; \bs(v_i, s_{-i}), s_{-i}(v_{-i}))\big]\notag \\
    \geq &~ 0 - \E_{v_{-i}\sim \D_{-i}}\big[u_i(v_i; \bs(v_i, s_{-i}), s_{-i}(v_{-i}))\big]\notag \\[0.8ex]
    = &~ \E_{v_{-i}\sim \D_{-i}}\big[(-v_i + \bs(v_i, s_{-i}))\Gamma_i(\bs(v_i, s_{-i}), s_{-i}(v_i))\big] \geq -\eps/10. \label{eq:step1-2}
\end{align}
Combining Eq.~\eqref{eq:step1-1} and Eq.~\eqref{eq:step1-2}, we have proved that
\[
\E_{v_{-i}\sim \D_{-i}}\big[u_i(v_i; s_i(v_i), s_{-i}(v_{-i})) - u_i(v_i; \bs(v_i, s_{-i}), s_{-i}(v_{-i}))\big] \geq \eps_i(v_i) - \frac{\eps}{10}.
\]
Taking an expectation over $v_i$, we complete the proof of the lemma.
\end{proof}

\vspace{+2mm}
{\noindent \bf Step 2: Rounding distribution. \ \ } 
Given a first-price auction $(\mN, \mB_{\eps}, \D, \Gamma)$, we would like to round the value distribution $\D_i$ such that it is supported over discrete values, and we truncate off the probability mass if it is too small. 
Formally, letting $\delta \in (0, \frac{\eps}{20})$ be a parameter to be specified later\footnote{Looking head, the parameter $\delta$ shall be much smaller than $\eps$}, we define $\D_{i}^{\eps, \delta}$ for each $i\in [n]$ as follows:
\begin{align*}
p_i^{\eps, \delta}(t) = \Pr_{v_i \sim \D_i^{\eps, \delta}}\left[v_i = \frac{t\eps}{10}\right] = 
\max\left\{0,\hspace{0.06cm}\int_{\frac{t\eps}{10}}^{\frac{(t+1)\eps}{10}}  \Pr_{v_i \sim \D_i}[v_i = v]\,\mathsf{d} v - \delta \right\} 
\end{align*}
for each $t\in [10\eps^{-1}-1]$,
and define
\begin{align*}
p_i^{\eps, \delta}(0) = \Pr_{v_i \sim \D_i^{\eps, \delta}}\big[v_i = 0\big] = 1- \sum_{t=1}^{10\eps^{-1}-1} \Pr_{v_i \sim \D_i^{\eps, \delta}}\left[v_i = \frac{t\eps}{10}\right].
\end{align*}
That is, the valuation distribution is rounded (down) to discrete values $V_{\eps} := \{0,  {\eps}/{10}, \ldots, 1 -  {\eps}/{10} \}$ and truncated at $\delta$; the extra probability mass is put on $0$.
From now on, we would consider both continuous and discrete valuation distribution of bidders.

Let $\D^{\eps, \delta} = \D_{1}^{\eps, \delta} \times \cdots \times \D_{n}^{\eps, \delta}$.
Next we define the notion of 
  $(\eps, \delta)$-approximate BNE --- a weaker notation of equilibrium with relaxed constraint on overbidding,
and show that it suffices to 
find an $(\eps,\delta)$-approximate BNE of the rounded FPA $(\mN, \mB_{\eps}, \D^{\eps, \delta}, \Gamma)$.

\begin{definition}[$(\eps, \delta)$-approximate Bayesian Nash equilibrium]
Let $n, m \geq 2$. Given a Bayesian FPA ($\mN, \mB_\eps, \D^{\eps, \delta},\Gamma$), a strategy profile $s = (s_1, \ldots, s_{n})$ is said to be an $(\eps, \delta)$-approximate Bayesian Nash equilibrium ($(\eps, \delta)$-approximate BNE), if for any player $i\in [n]$, its strategy $s_i$ is monotone and at most $\eps$-worse than the best response:  
\begin{align*}
    \E_{v \sim \D}\big[u_i(v_i; s_i(v_i), s_{-i}(v_{-i}))\big] \geq \E_{v\sim \D}\big[u_i(v_i; \bs(v_i, s_{-i}), s_{-i}(v_{-i}))\big] - \eps, 
\end{align*} 
Moreover, letting $v_{i, \max} = \max_{v_i \in V_{\eps}, p_i(v_{i, \max}) > 0} v_i$, the player $i$ never bids higher than $v_{i, \max}$ and its total overbidding probability is at most $\delta$, i.e.,
\begin{align*}
\Pr_{v_i\sim \calD_i^{\eps,\delta}}\big[s_i(v_i) > v_{i, \max}\big] = 0\quad\text{and}\quad\Pr_{v_i\sim \calD_i^{\eps,\delta}}\big[s_i(v_i) > v_i\big] < \delta.
\end{align*}
\end{definition}

We prove that it suffices to find an $(\eps, \delta)$-approximate BNE in the rounded FPA $(\mN, \mB_{\eps}, \D^{\eps, \delta}, \Gamma)$. 
In the proof we will perform the operation of converting a bidding strategy $s_i$ over bidding space $\mB_\eps$ and value distribution $\calD_i$ to a (unique) monotone bidding strategy $s_i'$ over the same bidding space $\mB_\eps$ but a different value distribution $\calD_i'$ such that their induced distributions over $\mB_\eps$ are the same:

\begin{definition}[Monotone analogue]
Let $s=(s_1,\ldots,s_n)$ be a strategy profile over bidding space $\mB_\eps$ and  value distribution $\D=\D_1\times\cdots\times \D_n$, and let $\D'=\D_1'\times \cdots\times \D_n'$ be another value distribution. The \emph{monotone analogue} of $s$  with respect to  $\D'$ (denoted as $\ms$) is defined as the unique monotone strategy profile $\ms=(\ms_1,\ldots,\ms_n)$ such that
\begin{align*}
\Pr_{v_i\sim \D_i'}\big[\ms_i(v_{i}) = b\big] = \Pr_{v_i\sim \D_i}\big[s_i(v_{i}) = b\big], \quad \text{for all }i\in [n]\ \text{and}\  b\in \mB_\eps.
\end{align*}
\end{definition}

\begin{lemma}
\label{lem:ptas-step2}
Given any FPA $(\mN, \mB_{\eps}, \D, \Gamma)$, let $\D^{\eps, \delta}$ be the rounded distribution defined above. Then the monotone analogue of any $(\eps, \delta)$-approximate BNE in $(\mN, \mB_{\eps}, \D^{\eps, \delta}, \Gamma)$ is an $(2\eps + 22\delta/\eps)$-approximate BNE in $(\mN, \mB_{\eps}, \D, \Gamma)$.
\end{lemma}
\begin{proof}
For clarity of the presentation, we divide the proof into two parts. 

{\bf First part. \ } Define the distribution $\D_{i}^{\eps}$ for each $i\in [n]$ as follow:
\begin{align*}
p_i^{\eps}(t) = \Pr_{v_i \sim \D_i}\left[v_i = \frac{t\eps}{10}\right] = 
\int_{\frac{t\eps}{10}}^{\frac{(t+1)\eps}{10}}  \Pr_{v_i \sim \D_i}[v_i = v]\,\mathsf{d} v,  \quad  \text{ for each }t \in [10\eps^{-1} - 1]
\end{align*}
and 
\begin{align*}
p_i^{\eps}(0) = \Pr_{v_i \sim \D_i}\big[v_i = 0\big] = 1- \sum_{t=1}^{10\eps^{-1}-1} \Pr_{v_i \sim \D_i}\left[v_i = \frac{t\eps}{10}\right].
\end{align*}
The distribution $\D^{\eps}_i$ rounds down $\D_i$ to discrete values $V_{\eps}$.
Given any $\eps$-approximate BNE strategy $s$ of FPA $(\mN, \mB_{\eps}, \D^{\eps}, \Gamma)$, we  show that its monotone analogue $\ms$ (with respect to $\D$) is an $\frac{11}{10}\eps$-approximate BNE of FPA $(\mN, \mB_{\eps}, \D, \Gamma)$. For any $v\in [0,1]$, let $v^{\eps} = \lfloor\frac{10v}{\eps}\rfloor\cdot \frac{\eps}{10}$ be the closest while smaller discrete value in $V_{\eps}$.  For any player $i\in [n]$, we have
\begin{align*}
 \E_{v_{-i}\sim \D_{-i}} &\big[u_i(v_i; \ms_i(v_i), \ms_{-i}(v_{-i}))\big] - \E_{v_{-i}\sim \D_{-i}}\big[u_i(v_i; \bs(v_i, \ms_{-i}), \ms_{-i}(v_{-i}))\big] \\
\geq &~\E_{v_{-i}\sim \D_{-i}}\big[u_i(v_i^{\eps}; \ms_i(v_{i}), \ms_{-i}(v_{-i}))\big] - \E_{v_{-i}\sim \D_{-i}}\big[u_i(v_i^{\eps}; \bs(v_i, \ms_{-i}), \ms_{-i}(v_{-i}))\big] - \frac{\eps}{10}\\
= &~\E_{v_{-i}\sim \D_{-i}^{\eps}}\big[u_i(v_i^{\eps}; \ms_i(v_{i}), s_{-i}(v_{-i}))\big] - \E_{v_{-i}\sim \D_{-i}^{\eps}}\big[u_i(v_i^{\eps}; \bs(v_i, s_{-i}), s_{-i}(v_{-i}))\big] - \frac{\eps}{10}\\
\geq &~ \E_{v_{-i}\sim \D_{-i}^{\eps}}\big[u_i(v_i^{\eps}; \ms_i(v_{i}), s_{-i}(v_{-i}))\big] - \E_{v_{-i}\sim \D_{-i}^{\eps}}\big[u_i(v_i^{\eps}; \bs(v_i^{\eps}, s_{-i}), s_{-i}(v_{-i}))\big] - \frac{\eps}{10}.
\end{align*}
The first step holds since 
$u_i(v_i; \ms_i(v_i), \ms_{-i}(v_{-i})) \geq u_i(v_i^{\eps}; \ms_i(v_i), \ms_{-i}(v_{-i}))$ ($v_i \geq v_i^{\eps}$)
and 
\[
u_i(v_i; \bs(v_i, \ms_{-i}), \ms_{-i}(v_{-i})) - u_i(v_i^{\eps}; \bs(v_i, \ms_{-i}), \ms_{-i}(v_{-i})) \leq (v_i - v_i^{\eps}) \leq \frac{\eps}{10}.
\]
The second step holds because the bidding profile $s_{-i}(v_{-i})$ when $v_{-i}\sim \D_{-i}^{\eps}$ is identical to  $\ms_{-i}(v_{-i})$ when $v_{-i}\sim \D_{-i}$.

Taking an expectation over $v_i \sim \D_i$, we obtain
\begin{align*}
 \E_{v\sim \D} &[u_i(v_i; \ms_i(v_i), \ms_{-i}(v_{-i}))] - \E_{v\sim \D}\big[u_i(v_i; \bs(v_i, \ms_{-i}), \ms_{-i}(v_{-i}))\big] \\[0.6ex]
\geq &~\E_{v_i \sim \D_i}\E_{v_{-i}\sim \D_{-i}^{\eps}}\big[u_i(v_i^{\eps}; \ms_i(v_{i}), s_{-i}(v_{-i}))\big] - \E_{v_i \sim \D_i} \E_{v_{-i}\sim \D_{-i}^{\eps}}\big[u_i(v_i^{\eps}; \bs(v_i^{\eps}, s_{-i}), s_{-i}(v_{-i}))\big] - \frac{\eps}{10}\\
= &~ \E_{v_i \sim \D_i^{\eps}}\E_{v_{-i}\sim \D_{-i}^{\eps}}\big[u_i(v_i; \ms_i(v_{i}), s_{-i}(v_{-i}))\big] - \E_{v_i \sim \D_i^{\eps}} \E_{v_{-i}\sim \D_{-i}^{\eps}}\big[u_i(v_i; \bs(v_i, s_{-i}), s_{-i}(v_{-i}))\big] - \frac{\eps}{10}\\
\geq &~ -\eps - \frac{\eps}{10} = -\frac{11}{10}\eps.
\end{align*}
where the second step holds due to $v_{i}^{\eps} (v_i \sim \D_i)$ has identical distribution as $v_i (v_i \sim \D_{i}^{\eps})$.
It is easy to verify that $\ms$ does not overbid given that $s$ does not overbid as an $\eps$-approximate BNE. We conclude that $\ms$ is an $\frac{11}{10}\eps$-approximate BNE of FPA $(\mN, \mB_{\eps}, \D^{\eps}, \Gamma)$.

{\bf Second part. \ } We prove that, given any $(\eps, \delta)$-approximate BNE $s$ of FPA $(\mN, \mB_{\eps}, \D^{\eps, \delta}, \Gamma)$, its monotone analogue $\ms$ with respect to the distribution $\calD^\eps$ must be  an $(\eps + 20\delta/\eps)$-approximate BNE of $(\mN, \mB_{\eps}, \D^{\eps}, \Gamma)$.

To prove the correctness, one needs to verify for each player $i\in [n]$,
\begin{enumerate}
    \item $\ms_i$ is monotone,
    \item $\ms_i$ is at most $(\eps + 20\delta/\eps)$-worse than the best response; and
    \item $\ms_i$ does not overbid.
\end{enumerate} 

The first claim of monotonicity follows directly from the definition of $\ms$ as the monotone~analogue of $s$. 
To prove the second claim, a critical observation is that the bidding profile is always preserved, i.e., $s(v)$ ($v\sim \D^{\eps, \delta}$) and $\ms(v)$ ($v\sim \D^{\eps}$) are identical. 
As a consequence, we have
\begin{align}
\E_{v\sim \D^{\eps}}\big[u_i(v_i; \ms_i(v_i), \ms_{-i}(v_{-i}))\big] = &~ \E_{v \sim \D^{\eps}}\big[(v_i - \ms_i(v_i))\Gamma_{i}(\ms_i(v_i), \ms_{-i}(v_{-i}))\big]\notag\\
\geq &~ \E_{v \sim \D^{\eps,\delta}}\big[(v_i - s_i(v_i))\Gamma_{i}(s_i(v_i), s_{-i}(v_{-i}))\big]\notag\\
= &~\E_{v \sim \D^{\eps, \delta}}\big[u_i(v_i; s_i(v_i), s_{-i}(v_{-i}))\big]. \label{eq:step2-1}
\end{align}
The second step holds since (i) the expected payment is the same under $\ms$ and $s$, i.e.,
\begin{align*}
 \E_{v \sim \D^{\eps}}\big[\ms_i(v_i)\Gamma_{i}(\ms_i(v_i), \ms_{-i}(v_{-i}))\big] = \E_{v \sim \D^{\eps,\delta}}\big[s_i(v_i)\Gamma_{i}(s_i(v_i), s_{-i}(v_{-i}))\big] 
\end{align*}
and (ii) the distribution $\D_i^{\eps}$ stochastically dominates the distribution $\D_i^{\eps,\delta}$.

Similarly, we have
\begin{align}
\E_{v\sim \D^{\eps}}\big[u_i(v_i; \bs(v_i, \ms_{-i}), \ms_{-i}(v_{-i}))\big] = &~ \E_{v \sim \D^{\eps}}\big[(v_i - \bs(v_i, \ms_{-i}))\Gamma_{i}(\bs(v_i, \ms_{-i}), \ms_{-i}(v_{-i}))\big]\notag\\
\leq &~ \E_{v \sim \D^{\eps, \delta}}\big[(v_i - \bs(v_i, s_{-i}))\Gamma_{i}(\bs(v_i, s_{-i}), s_{-i}(v_{-i}))\big] + 20\delta/\eps \notag \\
= &~ \E_{v\sim \D^{\eps, \delta}}\big[u_i(v_i; \bs(v_i, s_{-i}), s_{-i}(v_{-i}))\big]+ 20\delta/\eps \label{eq:step2-2}.
\end{align}
where the second step holds since the TV distance between $\D^{\eps}$ and $\D^{\eps, \delta}$ is bounded by $20\delta/\eps$ and $\ms_{-i}(v_{-i})$ ($v_{-i}\sim \D_{-i}^{\eps}$) has the same distribution as $s_{-i}(v_{-i})$ ($v_{-i}\sim \D_{-i}^{\eps, \delta}$).
Combining Eq.~\eqref{eq:step2-1} and Eq.~\eqref{eq:step2-2}, we conclude the $\ms$ is at most $(\eps +20\delta/\eps)$-worse than the best response.

Finally we prove the last item that $\ms_i$ does not overbid. 
We prove by contradiction and suppose that at some bid $b \in \mB_{\eps}$, 
we have $\Pr_{v_i \sim \D^{\eps}_i}[\ms_i(v_i) = b > v_i] > 0$. Let $\nu_i \in [0, 1]$ be the smallest value such that $\Pr[\ms_i(\nu_i) = b] > 0$, and let  $v_{i, \max}^{\eps, \delta} = \max_{v_i \in V_{\eps}, p_i^{\eps, \delta}(v_{i, \max}^{\eps, \delta}) > 0} v_i$.
We note that $\nu_i < b \leq v_{i, \max}^{\eps,\delta}$ due to the definition of $(\eps, \delta)$-approximate BNE, i.e., there is no overbidding over $v_{i,\max}^{\eps,\delta}$. Hence, we have
\begin{align}
\Pr_{v_i \sim \D^{\eps, \delta}}\big[s_i(v_i) \geq b\big] = &~ \Pr_{v_i \sim \D^{\eps}}\big[\ms_i(v_i) \geq b\big] \geq \Pr_{v_i \sim \D^{\eps}}\big[v_i \geq \nu_{i}\big] \geq \Pr_{v_i \sim \D^{\eps, \delta}}\big[v_i \geq \nu_{i}\big] + \delta\notag\\
\geq &~ \Pr_{v_i \sim \D^{\eps, \delta}}\big[v_i \geq b\big] + \delta \label{eq:step2-3}
\end{align}
The first step follows from the definition of monotone analogue, the second step
holds since $\ms_i$ is monotone. The third step holds since $\nu_i < v_{i, \max}^{\eps, \delta}$ and $\D^{\eps,\delta}$ truncates at least $\delta$ probability mass at $v_{i, \max}^{\eps, \delta}$.
Eq.~\eqref{eq:step2-3} implies 
\[
\Pr_{v_i \sim \D^{\eps, \delta}}[v_i < b \wedge s_i(v_i) \geq b] \geq\delta,
\]
i.e., the overbidding probability is at least $\delta$. This contradicts with the fact that $s$ is an $(\eps, \delta)$-approximate BNE of FPA $(\mN, \mB_{\eps}, \D^{\eps}, \Gamma)$.
We complete the proof here.
\end{proof}

\vspace{+2mm}
{\noindent \bf Step 3: Existence of discretized $(\eps,\delta)$-approximate BNE. \ \ }
Given a first-price auction $(\mN, \mB, \D, \Gamma)$, we prove the existence of a (suitably) discretized $(\eps, \delta)$-approximate BNE.
We describe this step using a generic FPA $(\mN, \mB, \D, \Gamma)$ but will apply it on the rounded FPA $(\mN, \mB_{\eps}, \D^{\eps, \delta}, \Gamma)$ later.
For any $j \in [0:m]$, $i \in [n]$, let $p_{i, j} = \Pr[s_i(v_i) = b_j]$ be the probability of player $i$ bidding $b_j$ in a strategy profile $s$.

\begin{lemma}[Discretization]
\label{lem:ptas-step3} Let $m, n \geq 2$, given any first-price auction $(\mN, \mB, \D, \Gamma)$ and $\mB = \{b_0, b_1, \ldots, b_m\}$ with $0 = b_0 < b_1 < \cdots < b_m \leq 1$, 
there exists an $(\eps, \delta)$-approximate BNE strategy profile $s$ such that $p_{i, j}$ is a integer multiple of $$\ell(m)\cdot \frac{1}{ \eps^{6}\delta  }$$ for any $i \in [n]$ and $ j \in [0:m]$, where $\ell(m)$ is some exponential function of $m$.
\end{lemma}

We make use of the following result from \cite{daskalakis2015approximate}. We note this the major part that we need a uniform tie-breaking rule.
\begin{theorem}[Theorem 3 of \cite{daskalakis2015approximate}]
\label{thm:dis}
Let $p_i \in \Delta_{m+1}$ for $i\in [n]$, and let $\{X_i \in \R^{m+1}\}_{i\in[n]}$ be a set of independent $(m+1)$-dimensional random unit vectors such that, for all $i \in [n]$, $j \in [0:m]$, $\Pr[X_i = \mathsf{1}_{j}] = p_{i,j}$.  Let $z > 0$ be an integer. Then there exists another set of probability vectors $\{\hat{p}_{i}\in \Delta_{m+1}\}_{i \in [n]}$ such that the following conditions hold:
\begin{itemize}
\item $|\hat{p}_{i, j} - p_{i, j}| = O( {1}/{z})$, for all $i \in [n]$ and $ j \in [0:m]$;
\item $\hat{p}_{i,j}$ is an integer multiple of $\frac{1}{2^{m}} \cdot \frac{1}{z}$ for all $i \in [n]$ and $ j \in [0:m]$;
\item If $p_{i, j} = 0$ then $\hat{p}_{i, j} = 0$;
\item Let $\{\hat{X}_i \in \R^{m+1}\}_{i\in[n]}$ be a set of independent $(m+1)$-dimensional random unit vectors\\ such that $\Pr[X_i = \mathsf{1}_{j}] = \hat{p}_{i,j}$ for all $i \in [n]$, $j \in [0:m]$.Then
\begin{align}
\left\|\sum_{i\in [n]}X_i - \sum_{i\in [n]}\hat{X}_i\right\|_{\mathsf{TV}} = O\left(h(m)\cdot \frac{\log z}{z^{1/5}}\right). \label{eq:pro1}
\end{align}
Moreover, for all $i' \in [n]$, we have 
\begin{align}
\left\|\sum_{i \in [n]\setminus \{i'\}}X_{i} - \sum_{i \in [n]\setminus \{i'\}}\hat{X}_{i}\right\|_{\mathsf{TV}} = O\left(h(m)\cdot \frac{\log z}{z^{1/5}}\right). \label{eq:pro2}
\end{align}
where $h(m)$ is some exponential function
\end{itemize}
\end{theorem}

\begin{proof}[Proof of Lemma \ref{lem:ptas-step3}]
Given a BNE strategy profile $s$ of FPA $(\mN, \mB, \D, \Gamma)$, we take 
\[
z = \Omega\Big(\max\big\{2h(m)^6 \eps^{-6}, 2m^2\delta^{-1}\big\}\Big) \quad \text{and} \quad \gamma =  h(m)\cdot \frac{\log z}{z^{1/5}}.
\]
Let $\{p_{i} \in \Delta^{m+1}\}_{i \in [n]}$ be the probability vectors that correspond to $s$, i.e., $p_{i,j}=\Pr[s_i(v_i)=b_j]$.
Using Theorem \ref{thm:dis}, let $\{\hat{p}_{i} \in \Delta^{m+1}\}_{i \in [n]}$ be the set of discretized probability vectors, and let $\hat{s}$ be the unique monotone strategy determined by $\hat{p}$ (with respect to the same value distribution $\calD$).
We prove that $\hat{s}$ forms an $(\eps, \delta)$-approximate BNE of the auction. We need to verify that for every player $i\in [n]$, (1) $\hat{s}_{i}$ is at most $\eps$-worse than the best response; and (2) the overbidding probability is small.

By Eq.~\eqref{eq:pro2}, for any $j \in [0:m]$, one has 
\begin{align}
\E_{v_{-i}\sim \D_{-i}}\big[\Gamma_i(b_j, s_{-i}(v_{-i})) \big]= \E_{v_{-i}\sim \D_{-i}}\big[\Gamma_i(b_j, \hat{s}_{-i}(v_{-i}))\big] \pm \gamma.\label{eq:step3-1}
\end{align}
since the allocation probability (under the uniform tie-breaking) is determined by the histogram of other players' bidding histogram, which shifts by at most $\gamma$ between $s_{-i}$ and $\hat{s}_{-i}$ in total variance distance. 

For the utility, we have
\begin{align}
\E_{v \sim \D} \big[u_{i}(v_i; \hat{s}_i(v_i), \hat{s}_{-i}(v_{-i}))\big] = &~ \E_{v \sim \D} \big[(v_{i} - \hat{s}_i(v_i)) \cdot \Gamma_i(\hat{s}(v_i), \hat{s}_{-i}(v_{-i}))\big]\notag\\
= &~ \E_{v \sim \D} \big[(v_{i} - \hat{s}_i(v_i)) \cdot \Gamma_i(\hat{s}_i(v_i), s_{-i}(v_i))\big] \pm \gamma \notag\\
= &~ \E_{v \sim \D} \big[(v_{i} - s_i(v_i)) \cdot \Gamma_i(s_i(v_i), s_{-i}(v_i))\big] \pm O\left(m^2\cdot \frac{1}{z}+\gamma\right) \notag \\
= &~ \E_{v \sim \D} \big[u_{i}(v_i; s(v_i), s_{-i}(v_{-i}))\big] \pm \frac{\eps}{2}.\label{eq:step3-2}
\end{align}
The second step follows from Eq.~\eqref{eq:step3-1}, the third step holds since the TV distance between $(v_i, s_i(v_i))$ and $(v_i, \hat{s}_i(v_i))$ is at most $O(m^2 \cdot \frac{1}{z})$.
The last step follows from the choice of $z$.

At the same time, we have
\begin{align}
\E_{v \sim \D} \big[u_{i}(v_i; \bs(v_i, \hat{s}_{-i}), \hat{s}_{-i}(v_{-i}))\big] \leq &~ \E_{v \sim \D} \big[u_{i}(v_i, \bs(v_i, \hat{s}_{-i}), s_{-i}(v_{-i}))\big] +\gamma \notag \\
\leq &~ \E_{v \sim \D} \big[u_{i}(v_i, \bs(v_i, s_{-i}), s_{-i}(v_{-i}))\big] +\gamma \notag \\
\leq &~ \E_{v \sim \D} \big[u_{i}(v_i; \bs(v_i, s_{-i}), s_{-i}(v_{-i})) \big]+\frac{\eps}{2}. \label{eq:step3-3}
\end{align}
Here the first step follows from the bidding histogram shifts by at most $\gamma$ between $s_{-i}$ and $\hat{s}_{-i}$ in total variance distance. The last step follows from the choice of parameters.

Combining Eq.~\eqref{eq:step3-2} and Eq.~\eqref{eq:step3-3}, we have that $\hat{s}_i$ is at most $\eps$-worse than the best response.

To bound the probability of overbidding, let $j(i) \in [0:m]$ be the maximum bid that receives non-zero probability under $\hat{s}_i$, then it is easy to verify that  $b_{j(i)} \leq v_{i,\max}$: otherwise $b_{j(i)}>v_{i,\max}$ and $p_{i, j(i)}>0$ would contradict with the no overbidding assumption of BNE on $s$.  

Finally we have
\begin{align*}
\Pr_{v_i\sim\calD_i}\big[\hat{s}_i(v_i) > v_i\big] \leq \Pr_{v_i\sim\calD_i}\big[ {s}_i(v_i) > v_i\big] + O\left(m^2 \cdot \frac{1}{z}\right) = 0 + O\left(m^2 \cdot \frac{1}{z}\right) \leq \delta
\end{align*}
since the total variation distance between $(v_i, s_i(v_i))$ and $(v_i, \hat{s}_i(v_i))$ is at most $O(m^2 \cdot \frac{1}{z})$.
\end{proof}



\vspace{+2mm}
{\noindent \bf Step 4: Searching for an $(\eps,\delta)$-approximate BNE. \ \ } 
Finally we provide a simple searching algorithm for $(\eps,\delta)$-approximate BNE over the discretized space. The key observation comes from the allocation rule of a first-price auction, i.e., only players with the highest bid could win the item. 
We say a strategy profile $s$ lies in the grid $S_{\omega}$ for some $\omega \in (0,1)$ if $p_{i, j} = \Pr_{v_i \sim \D_i}[s_i(v_i) = b_j]$ is a multiple of $\omega$ for every $i \in [n]$ and $j \in [0:m]$.

\begin{lemma}
\label{lem:ptas-step4}
Given a first-price auction $(\mN, \mB, \D, \Gamma)$, and suppose there exists at least one $(\eps, \delta)$-approximate BNE over the grid $S_{\omega}$, then 
there is an algorithm that runs in $n^4m \cdot 2^{\tilde{O}(m/\eps\omega)}$ time and returns an $(2\eps, \delta)$-approximate BNE under the uniform tie-breaking rule.
\end{lemma}
\begin{proof}
Let $R = 100(m+1)/\eps\omega$. The discretized strategy profiles $\sE_{\omega} \subseteq [\Delta_{m+1}]^{R}$ are defined as follows. A strategy profile $(p_1, \ldots, p_{R}) \in \sE_{\omega}$ is parameterized by a bid level $j^{*} \in [0:m-1]$ and  $k_0,k_1, \ldots, k_{j^{*}} \in [0:10\eps^{-1} ]$ such that $p_{i, j}$ is a multiple of $\omega$ for all $i \in [n], j \in [0:m]$, and
\begin{align*}
\sum_{r=1}^{R} p_{r, m - j} = k_j,\, \forall j \in [0: j^{*}] \quad \text{and} \quad p_{r, m - j} = 0, \, \forall r\in [R],j \in [j^{*}+ 1: m-1].
\end{align*}

The size of $\sE_{\omega}$ satisfies
\[
|\sE_{\omega}| \leq (m+1) \cdot (10/\eps\omega + 1/\omega)^{m+1} \cdot \binom{R}{10/\eps\omega + 1/\omega} \leq 2^{\tilde{O}(m/\eps\omega)}.
\]

For any strategy profile $(p_1, \ldots, p_R) \in \sE_{\omega}$, one can augment with $(n-R)$ default strategies $(1,0, \ldots, 0)$ (that is, bidding $b_0=0$ with probability $1$). 
Slightly abuse of notation, we also use $\sE_{\omega} \subseteq [\Delta_{m+1}]^{n}$ to denote the augmented strategy profiles. We shall prove 
\begin{itemize}
\item There is an $(2\eps, \delta)$-approximate BNE in $\sE_{\omega}$ (up to a matching with players), and
\item One can identify the matching efficiently.
\end{itemize}

\paragraph{Existence of $(2\eps, \delta)$-approximate BNE} Recall that there exists an $(\eps, \delta)$-approximate BNE strategy $s$ over the grid $S_{\omega}$. Let $j^{*}$ be the first index over $[0:m]$ such that $\sum_{i\in [n]} p_{i,m-j} \leq 10/\eps$ ($\forall j < j^{*}$) and $\sum_{i\in [n]} p_{i,m- j^{*}} \geq 10/\eps$.
If $j^{*} = m$, then we have $s \in \sE_{\omega}$ (up to a matching between players).
If $j^{*} \leq m-1$, then let $n^{*} \in [n]$ be the smallest player such that $\sum_{i\in [n^{*}]} p_{n,m-j^{*}} \geq 10/\eps$ (w.l.o.g. we assume it takes the equality). 
Truncate the strategy profile to $s'$ such that
\begin{align*}
p'_{i, j} = \Pr_{v_i\sim \D_i}[s_i'(v_i) = b_j] = \left\{
\begin{matrix}
p_{i, j} & j > m - j^{*} \vee (j = m - j^{*} \wedge i \leq n^{*})\\
0 & j \in [1: m - j^{*} - 1] \vee (j = m - j^{*} \wedge i > n^{*})\\
1 - \sum_{j' =1}^{m}p'_{i, j'} & j = 0.
\end{matrix}
\right.
\end{align*}

It is clear the new strategy $s' \in \sE_{\omega}$ (up to a matching between players), and for each player $i\in [n]$, the new strategy $s_i'$ is monotone and the probability of overbidding is no more than $\delta$ (because we only move bidding probability to $b_0 = 0$). It suffices to prove it is at most $2\eps$-worse than the best response. 
The key observation is that the allocation probability of bidding $b_{j}$ with $j > m- j^{*}$ remains the same, i.e., 
\begin{align*}
\E_{v_{-i}\sim \D_{-i}}[\Gamma_{i}(b_j, s_{-i}'(v_{-i}))] = \E_{v_{-i}\sim \D_{-i}} [\Gamma_{i}(b_j, s_{-i}(v_{-i}))] \quad \forall j \in [m- j^{*} + 1: m]
\end{align*}
and moreover, the allocation probability of bidding no more than $b_{m-j^{*}}$ is small
\begin{align*}
\E_{v_{-i}\sim \D_{-i}}[\Gamma_{i}(b_j; s_{-i}'(v_{-i}))] \leq \eps/4 \quad \forall j \in [0: m - j^{*}].
\end{align*}
This holds since with probability at least $1-\exp(-5/3\eps)$, there are at least $5/\eps$ players bid no less than $b_{m-j^{*}}$ by Chernoff bound.

Therefore, at any value point $v_i$, if $s_i'(v_i) > b_{m- j^{*}}$, then $s_i(v_i) = s_i'(v_i)$ and
\begin{align}
\E_{v_{-i}\sim \D_{-i}}[u_{i}(v_i; s_i'(v_i); s_{-i}'(v_{-i}))] = \E_{v_{-i}\sim \D_{-i}}[u_{i}(v_i; s_i(v_i); s_{-i}(v_{-i}))] \label{eq:step4-1}
\end{align}
If $s_i'(v_i) \leq b_{m- j^{*}}$, then $s_i(v_i) \leq b_{m- j^{*}}$, and 
\begin{align}
\E_{v_{-i}\sim \D_{-i}}[u_{i}(v_i; s_i'(v_i); s_{-i}'(v_{-i}))] \geq -\eps/4 \quad \text{and}\quad \E_{v_{-i}\sim \D_{-i}}[u_{i}(v_i; s_i(v_i); s_{-i}(v_{-i}))] \leq \eps/4. \label{eq:step4-2}
\end{align}
Combining Eq.~\eqref{eq:step4-1} and Eq.~\eqref{eq:step4-2}, for any $v_i \in [0, 1]$, one has
\begin{align}
\E_{v_{-i}\sim \D_{-i}}[u_{i}(v_i; s_i'(v_i); s_{-i}'(v_{-i})) - u_{i}(v_i; s_i(v_i); s_{-i}(v_{-i}))] \geq -\eps/2. \label{eq:step4-4}
\end{align}
Similarly, one can prove 
\begin{align}
\E_{v_{-i}\sim \D_{-i}}[u_i(v_i;\bs(v_i, s_{-i}'), s_{-i}'(v_{-i}))] \leq \E_{v_{-i}\sim \D_{-i}}[u_i(v_i;\bs(v_i, s_{-i}), s_{-i}(v_{-i}))]  + \eps/2. \label{eq:step4-3}
\end{align}
Combining Eq.~\eqref{eq:step4-4} \eqref{eq:step4-3}, we have proved $s_i'$ is at most $2\eps$-worse than the best response.

\paragraph{Find a matching}Given a strategy profile $(p_1, \ldots, p_n)\in \sE_{\omega}$, we show how to define a bipartite matching problem, such that the $(2\eps, \delta)$-approximate BNE are one-to-one correspondence to the perfect bipartite matching of the graph.
The bipartite matching problem is defined between players $[n]$ and the strategy $\{p_i\}_{i \in [n]}$. 
We draw an edge between player $i_1$ and the strategy $p_{i_2}$, if $p_{i_2}$ is at most $2\eps$-worse than the best response (note the histogram of other players' bidding profile is determined) and the overbidding probability is small at most $\delta$. 
We note the best response can be computed in $O(mn^2)$ time and one can estimate the utility of a bid in $O(n^2)$ time using dynamic programming (similar as \cite{filos2021complexity}). Hence, it takes $n^2 \cdot O(mn^2)$ time to construct the bipartite graph, and a perfect matching can be found in time $O(n^3)$.
\end{proof}

Combining the above four steps (Lemma~\ref{lem:ptas-step1}, Lemma~\ref{lem:ptas-step2}, Lemma~\ref{lem:ptas-step3} and Lemma~\ref{lem:ptas-step4}), we conclude the proof of Theorem \ref{thm:ptas}.


\section*{Acknowledgement}
X.C. and B.P. would like to thank Aviad Rubinstein and anonymous STOC reviewers for helpful suggestions on the paper. The research of X.C. and B.P. is supported by NSF grants CCF-1703925, IIS-1838154, CCF-2106429 and CCF-2107187, CCF-1763970, CCF-2212233, COLL2134095, COLL2212745.

\newpage
\bibliographystyle{alpha}
\bibliography{ref}

\newpage
\appendix
\section{Missing proof from Section \ref{sec:pre}}
We provide the proof of Lemma \ref{lem:notion}.
\begin{proof}[Proof of Lemma \ref{lem:notion}]
Given an $\eps$-approximate BNE strategy $s$, we transform each individual strategy $s_i$ into another strategy $s_i'$, such that
\begin{align}
\E_{v_{-i}\sim \D_{-i}}[u_i(v; s_i'(v_i), s_{-i}(v_{-i})) - u_i(v; \bs(v_i; s_{-i}),  s_{-i}(v_{-i}))] \leq 10\sqrt{\eps} \quad  \forall v\in [0,1], i \in [n], \label{eq:transform1}
\end{align}
and
\begin{align}
\Pr_{v_i \sim \D_i}[s_i'(v_i) \neq s_{i}(v_i)] \leq \sqrt{\eps}\quad \forall i \in [n].\label{eq:transform2}
\end{align}

Given the above properties, one can easily conclude that $s' = (s_1', \ldots, s_n')$ is an $\eps'$-BNE with $\eps' = 10\sqrt{\eps} + 2n\sqrt{\eps}$, since the bidding histogram changes at most $n\sqrt{\eps}$ in total variational distance.

Let $\bs_{\delta}(v_i, s_{-i})\subseteq \mB$ be the set of $\delta$-best response (note we slightly abuse notation and its meaning is different from Section \ref{sec:ptas}), i.e., 
\begin{align*}
\bs_{\delta}(v, s_{-i}):= \big\{b\in \mB: \E_{v_{-i}\sim \D_{-i}}\big[u_i(v; \bs(v, s_{-i}), s_{-i}(v_{-i})) - u_i(v; b, s_{-i}(v_{-i}))\big] \leq \delta, b\leq v \big\}.
\end{align*}

To obtain the desired strategy $s_i'$, we modify the strategy $s_i$ in a descending order, from $s_i(1)$ to $s_i(0)$.
It would be convenient to describe it in a continuous way, though we note it could be easily implemented in polynomial time.
For each value $v \in [0, 1]$, let $b^{+}(v)$ be the smallest bid given the private value is above $v$, i.e., $b^{+}(v):= \min_{b \in \mB, v' > v}s_i'(v')$.
We divide into cases.
\begin{enumerate}
    \item If $s_i(v) \in \bs_{4\sqrt{\eps}}(v; s_{-i})$ and $s_i(v) \leq b^{+}(v)$, then we keep the strategy unchanged, i.e., $s_{i}'(v) = s_{i}(v)$;
    \item Otherwise, let $B_{\delta}(v) := \bs_{\delta}(v_i, s_{-i}) \cap [0, b^{+}(v)]$ be the set of available $\delta$-best response at $v$, then
    \begin{enumerate}
    \item If $B_{4\sqrt{\eps}}(v) \neq \emptyset$, then take $s_i'(v) = \max_{b\in B_{5\sqrt{\eps}}(v)}b$; otherwise
    \item Take $s_i'(v) = b^{+}(v)$ to be the maximum possible bid.
    \end{enumerate}
\end{enumerate}
It is easy to verify that the strategy is monotone. There is no overbidding since in case (2-b), $b^{+}(v) \leq \bs(v, s_{-i}) \leq v$. 
We next verify the strategy $s'$ satisfies the aforementioned properties.

To prove Eq.~\eqref{eq:transform1}, it suffices to consider the value $v\in [0,1]$ such that $s_i'(v) = b^{+}(v)$ and  $b^{+}(v) < b_{\min}(v)$ (i.e., case 2-b), where $b_{\min}(v) =\min_{b \in \bs_{4\sqrt{\eps}}(v, s_{-i})}b$ .
Let $v^{+} > v$ be the largest value such that $s_i'(v^{+}) = b^{+}(v)$, the maximum exists because it is wlog to assume the strategy is right closed. 
Then we know that $b^{+}(v) \in \bs_{5\sqrt{\eps}}(v^{+}, s_{-i})$ (i.e., it must fall into case 1 or case 2-a), therefore, 
\begin{align*}
5\sqrt{\eps} \geq &~ \E[u_{i}(v^{+};b_{\min}(v), s_{-i}(v_{-i})) - u_{i}(v^{+}; s_i'(v^{+}), s_{-i}(v_{-i}))] \\
\geq &~ \E[u_{i}(v;b_{\min}(v), s_{-i}(v_{-i})) - u_{i}(v; s_i'(v^{+}), s_{-i}(v_{-i}))].
\end{align*}
Here the second step holds due to $v^{+} \geq v$, $b_{\min}(v) > s_i'(v^{+}) = b^{+}(v)$ and the monotonicity of the utility difference.

Therefore, we have
\begin{align*}
&~ \E[u_{i}(v;\bs(v, s_{-i}), s_{-i}(v_{-i})) - u_{i}(v; s_{i}'(v), s_{-i}(v_{-i}))] \\
= &~ \E[u_{i}(v;\bs(v, s_{-i}), s_{-i}(v_{-i})) -u_i(v;b_{\min}(v), s_{-i}(v_{-i})) +  u_i(v;b_{\min}(v), s_{-i}(v_{-i})) - u_{i}(v; s_i'(v^{+}), s_{-i}(v_{-i}))] \\
\leq &~ 4\sqrt{\eps} + 5\sqrt{\eps} = 9\sqrt{\eps}.
\end{align*}
The second step holds since $b_{\min}(v) \in \bs_{4\sqrt{\eps}}(v, s_{-i})$. This proves Eq.~\eqref{eq:transform1}.


To prove  Eq.~\eqref{eq:transform2}, let $V_i = \{v: s_i(v) \notin \bs_{\sqrt{\eps}}(v; s_{-i})\}$. By Markov inequality, the probability on $V_i$ is small, i.e., 
\[
\Pr[v\in V_i] \leq \sqrt{\eps}.
\]

We next prove $s_i'(v) = s_i(v)$ holds for any $v \in [0,1]\backslash V_i$.
We prove by contradiction and suppose Eq.~\eqref{eq:transform2} is violated at value $v \in [0,1]\backslash V_i$. 
Then we have $s_i(v) > b^{+}(v)$ (i.e., not case 1).
Let $v^{+} > v$ be the maximum value such that $s_i'(v^{+}) < s_{i}(v)$, again the maximum exists because it is wlog to assume the strategy is right closed.
At value $v^{+}$, the bid $s_i(v)$ is feasible (i.e., $s_i(v) \leq b^{+}(v^{+})$) and we are not at case 2-b, due to the maximum assumption. 
It we are at case 1, then we have $s_i'(v^{+}) = s_{i}(v^{+}) \geq  s_{i}(v)$ due to monotonicity of $s$, this contradicts with the assumption that $s_i'(v^{+}) < s_{i}(v)$.
Therefore, the only possibility is that we are at Case 2-a, and it suffices to prove $s_i(v) \in \bs_{5\sqrt{\eps}}(v^{+}, s_{-i})$ to establish a contradiction. 

At case 2-a, let $b' \in B_{4\sqrt{\eps}}(v^{+})$ be any feasible $4\sqrt{\eps}$-best response. We have $b' \leq s_i'(v_1)$, and
\begin{align*}
&~\E[u_{i}(v^{+}; \bs(v^{+}; s_{-i}), s_{-i}(v_{-i})) - u_{i}(v^{+}; s_i(v), s_{-i}(v_{-i}))]\\
=&~\E[u_{i}(v^{+}; \bs(v^{+}; s_{-i}), s_{-i}(v_{-i})) - u_{i}(v^{+}; b', s_{-i}(v_{-i})) + u_{i}(v^{+}; b', s_{-i}(v_{-i})) - u_{i}(v^{+}; s_i(v), s_{-i}(v_{-i}))]\\
\leq &~4\sqrt{\eps} + \E[u_{i}(v^{+}; b', s_{-i}(v_{-i})) - u_{i}(v^{+}; s_i(v), s_{-i}(v_{-i}))]\\
\leq &~4\sqrt{\eps} + \E[u_{i}(v; b', s_{-i}(v_{-i})) - u_{i}(v; s_i(v), s_{-i}(v_{-i}))]
\leq 5\sqrt{\eps}.
\end{align*}
Here, the second step follows from $b' \in B_{4\sqrt{\eps}}(v^{+}) \subseteq \bs_{4\sqrt{\eps}}(v^{+}, s_{-i})$, the third step holds due to $v^{+} > v$, $b' \leq s_{i}(v)$ and the monotonicity of utility difference, 
and the last step follows from $v \notin V_1$. We finish the proof here.
\end{proof}

\end{document}